\documentclass[preprint,12pt,square]{elsarticle}

\usepackage{amssymb}
\usepackage{latexsym}
\usepackage{mathrsfs}
\usepackage{multirow}
\usepackage{graphicx}
\usepackage{subfig}
\usepackage{caption}
\usepackage{subcaption}
\usepackage{amsmath}
\usepackage{amsthm}
\usepackage{commath}
\usepackage{gensymb}
\usepackage{booktabs}
\usepackage{makecell}
\usepackage{hyperref}
\usepackage[capitalise]{cleveref}
\usepackage{xcolor}
\usepackage{lineno}
\usepackage{tabularx}
\definecolor{mygreen}{RGB}{255,165,0}

\journal{Journal of computational physics}

\begin{document}

\begin{frontmatter}



\title{An energy-stable phase-field model for droplet icing simulations}


\author[1]{Zhihua Wang}
\author[2]{Lijing Zhou}
\author[3]{Wenqiang Zhang}
\author[4]{Xiaorong Wang}
\author[1]{Shuguang Li}
\author[4]{Xuerui Mao\corref{cor1}}
\cortext[cor1]{Corresponding author.}\ead{xmao@bit.edu.cn}
\affiliation[1]{organization={Faculty of Engineering, University of Nottingham},
            city={Nottingham},
            postcode={NG7 2RD}, 
            country={UK}}
\affiliation[2]{organization={School of Aerospace Engineering, Beijing Institute of Technology},
            city={Beijing},
            postcode={100081},
            country={China}}
\affiliation[3]{organization={School of Mechatronical Engineering, Beijing Institute of Technology},
            city={Beijing},
            postcode={100081}, 
            country={China}}
\affiliation[4]{organization={Advanced Research Institute of Multidisciplinary Sciences, Beijing Institute of Technology},
            city={Beijing},
            postcode={100081}, 
            country={China}}   
\begin{abstract}
A phase-field model for three-phase flows is established by combining the Navier-Stokes (NS) and the energy equations, with the Allen-Cahn (AC) and Cahn-Hilliard (CH) equations and is demonstrated analytically to satisfy the energy dissipation law. A finite difference scheme is then established to discretize the model and this numerical scheme is proved to be unconditionally stable. Based on this scheme, the droplet icing process with phase changing is numerically simulated and the pointy tip of the icy droplet is obtained and analyzed. The influence of the temperature of the supercooled substrate and the ambient air on the droplet freezing process is studied. The results indicate that the formation of the droplet pointy tip is primarily due to the expansion in the vertical direction during the freezing process. Lower substrate temperatures can accelerate this process. Changes in air temperature have a relatively minor impact on the freezing process, mainly affecting its early stages. Moreover, our results demonstrate that the ice front transitions from an approximately horizontal shape to a concave one. Dedicated physical experiments were conducted and the measured solidification process matches the results of the proposed phase-field method very well.
\end{abstract}


\end{frontmatter}


\section{Introduction}
Icing has received significant attention in scientific research and engineering applications not only because of its influence on natural ecosystems and human activities but also due to its adverse impacts on engineering applications, e.g., road traffic, aeronautics, electrical grids, etc. \cite{akhtar2023comprehensive, tiwari2023droplet}.

On the positive side, icing plays a vital role in nature. Ice formations like glaciers and snowflakes are iconic features of Earth's landscapes. In aquatic environments, ice acts as an insulating barrier, preventing rapid heat loss from water and supporting aquatic life during winter. In medicine, icing therapy is commonly used to reduce swelling and accelerate recovery \cite{bleakley2004use}. Historically, ice has been used for cooling, with early examples dating back to 400 BC, when Persian engineers developed techniques for ice storage and cooling in the desert \cite{hosseini2012overview}. Nowadays, water freezing is not only for cooling but also for meteorology and climate change \cite{knopf2020stochastic}, spray freezing \cite{WANNING2015136, makkonen1987salinity}, ground freezing and thermal energy storage \cite{akhtar2023comprehensive, xu2022experimental}, to name a few. However, icing also poses significant challenges in engineering. Ice accumulation on roads creates hazardous conditions, while ice on aircraft wings reduces lift and increases drag, leading to safety concerns \cite{andersson2011impact, hatch2015national, valarezo1993maximum}. Icing on wind turbine blades can cause up to 50 per cent production loss and lead to structural damage \cite{lv2014bio, dalili2009review}. Thus, while icing offers natural benefits and practical uses, it also presents considerable challenges, particularly in fields where safety and efficiency are paramount. A crucial step in addressing these challenges is investigating the mechanism of the single droplet icing process since it is the very first stage of ice accretion \cite{yue2019freezing} and is of vital significance to the development of icing prediction, anti-icing and de-icing methods \cite{cebeci2003aircraft, thomas1996aircraft}.

One of the most intriguing phenomena of the single droplet icing is the formation of the pointy tip after freezing. Efforts have been therefore devoted to investigating the mechanism of the tip singularity formation. Snoeijer et al.~\cite{snoeijer2012pointy} revealed that the singular tip originates from the density reduction in the icing process and an analytical geometric model elucidating the correlation between the cone angle of the tip and the density ratio of ice and water was proposed. Marin et al.~\cite{marin2014universality} experimentally demonstrated that the cone angles of the tip are universal and independent of the temperature of the supercooled substrate and the solidification rate. Subsequently, experimental and theoretical analyses of the freezing droplets with different contact angles and temperatures \cite{zhang2019shape,tembely2019comprehensive}, different subcooled surfaces \cite{zhang2018experimental,stiti2020icing}, and various heights of droplets from the cold substrates \cite{hu2020frozen} were carried out successively. Effects of other parameters, such as gravity \cite{zeng2022influence}, air humidity \cite{sebilleau2021air} and substrates wettability on the droplet icing \cite{chang2023experimental,pan2019experimental}, the asymmetric cooling influence on the orientation of the freezing tips \cite{starostin2022effect}, the re-icing process on superhydrophobic surface \cite{chu2019droplet}, and numerous other experiments \cite{jin2010icing,zhao2021freezing,starostin2022universality,schremb2016solidification,jin2017experimental}, have enriched our insight into the droplet icing process and the formation of the icy tips.

Compared to the extensive experimental investigations, numerical studies on the mechanism of droplet icing with pointy tip formation are limited. Among them, the front-tracking method \cite{vu2015numerical}, the level-set method \cite{lian2017experimental}, the phase-field method \cite{hagiwara2017ice}, the volume-of-fluid method \cite{yao2018experimental,zhang2018simulation}, the improved smoothed particle hydrodynamics method \cite{zhang2022numerical}, the MDPDE method (many-body dissipative particle dynamics with energy conservation configurations) \cite{wang2020mesoscopic}, and enthalpy-based methods \cite{chaudhary2014freezing, wang2021numerical} have been applied in the simulation of droplet icing processes. Although some of the results agreed with the experimental ones and the pointy tips were captured \cite{vu2015numerical,zhang2018simulation,wang2020mesoscopic}, in most of the aforementioned studies \cite{chang2023experimental,hagiwara2017ice,yao2018experimental,zhang2022numerical,chaudhary2014freezing,wang2021numerical}, pointy tips were not obtained, demonstrating that numerically capturing the pointy tip of the droplet remains a challenge. Therefore, it is important to develop a numerical model to simulate the droplet freezing process, considering the influence of the surrounding air.

To study the icing process and the pointy tip formation in such a three-phase (ice, water, and air) flow, how to capture the interfaces between different phases is one of the critical issues. In the mid-20th century, the idea of a non-zero thickness interface was developed, and the phase-field method was proposed \cite{cahn1958free}. This method circumvents the imposition of boundary conditions on the interface by introducing additional variables continuous in the domain so that the physical variables do not exhibit mutations at the interface \cite{mirjalili2019comparison}. Taking this advantage, this method was then applied in the simulation of two-phase flows such as solidification \cite{kobayashi1993modeling, steinbach2009phase}, melting \cite{li2018phase, hester2020improved}, boiling \cite{juric1998computations, borcia2005phase}, and sublimation \cite{reitzle2019direct,kaempfer2009phase}, as well as three-phase flows \cite{zhang2022phase, huang2022consistent} by implementing the Cahn-Hilliard (CH) equation \cite{cahn1958free} and Allen-Cahn (AC) equation \cite{allen1979microscopic} simultaneously. Despite these successful applications, there is a lack of modelling of droplet icing with pointy tip formation based on the phase-field method. 

Therefore, the purpose of this study is to establish a three-phase phase-filed method and simulate the droplet icing process. In this method, two order parameters (denoted as $\phi$ and $c$ in the following) are introduced to describe different phases with interfaces. $\phi \in [-1, 1]$ is used to distinguish air ($\phi=-1$) from water ($\phi=1$), while $c \in [-1, 0]$ is employed to separate water ($c=0$) from ice ($c=-1$). It is worth mentioning that the definition of the range for $\phi$ and $c$ is arbitrary, for example, $c$ can be defined in the range $\left[ -1, +1 \right]$ \cite{yang2006numerical, shen2012modeling, shen2010phase}. In this paper, we refer to Ref. \cite{zhang2022phase} and define $c$ to be in the range $\left[ -1, 0 \right]$. The CH equation is applied to model the evolution of the interface between water and air due to the volume conservation of the whole domain while the AC equation is employed to model the phase change from water to ice since the influence of the latent heat is considered in a specific form of the AC equation \cite{zhang2022phase, boettinger2002phase}. By combining these phase-field equations (CH and AC equations) with the Navier-Stokes (NS) equations and an energy conservation equation, an AC-CH-NS phase-field model for three-phase flows is established.

The satisfaction of the energy dissipation law, indicating the energy stability of the methods, is one of the important properties of most phase-field models \cite{shen2012modeling, yang2023modified2, yang2023phase2}. Previous work has well studied this law for two-phase models, including the AC-NS and CH-NS frameworks \cite{shen2012modeling}. The satisfaction of the proposed AC-CH-NS model to this law is proved in the present work, demonstrating that the proposed scheme maintains numerical stability. To the best of our knowledge, this is the first study to show the energy dissipation of the AC-CH-NS phase-field model with a discrete energy-stable computational scheme.

The rest of the paper is organized as follows: in Section \ref{2}, a numerical phase-field model for three-phase flows combining the AC equation, the CH equation, the NS equations, and the energy conservation equation is established, before the proof of the energy dissipation law of the proposed model. In Section \ref{3}, an energy-stable computational algorithm including the temporal and spatial discretization schemes with finite difference method (FDM) is presented. In Section \ref{4}, the numerical results of droplet icing are presented and analysed with pointy tips observed at the top of the icy droplets. The effects of the substrate temperature and the ambient air are investigated by parameter studies. In Section \ref{5}, conclusions are drawn.

\section{Proposed model for three-phase flows} \label{2}
The Allen-Cahn equation \cite{allen1979microscopic} is used to model the phase change from water to ice:
\begin{equation}
    \frac{\partial c}{\partial t}+\boldsymbol u \cdot \nabla c=M_c[\xi_c ^2 \nabla^2c-F'(c)],
    \label{ac}
\end{equation}
where $\boldsymbol u$ denotes the velocity vector, $t$ is time, $\xi_c$ denotes the thickness of the ice-water interface, and $M_c$ is the phenomenological mobility coefficient, whose larger value speeds up the phase-changing process. Refer to \cite{zhang2022phase} for a systematic parameter study on $M_c$. $F'(c)$ denotes the derivative of $F(c)$ with respect to $c$, where the double-well potential, $F(c)$, with bistable solutions at both ice phase ($c = -1$) and water phase ($c = 0$), is defined as \cite{zhang2022phase, boettinger2002phase, kobayashi2010brief}: 
\begin{equation}
    	F(c)= H_1(c)+\frac{L}{W_a}\frac{T_M-T}{T_M}H_2(c), \label{fc}
\end{equation}
where $H_1(c)=\frac{1}{2}(c+1)^2 c^{2}$ gives a symmetric part of the potential. It is worth mentioning that $H_1(c)=\frac{1}{4}(c+1)^2 (c-1)^{2}$ was adopted in Refs \cite{yang2006numerical, shen2012modeling, huang2020consistent2} because of the different definition of the bistable solutions ($c=-1$ and $c=1$ in their work are defined as the bistable solutions at their defined two phases, respectively). $T$ is the temperature of the flow field, $T_M$ is the melting point of the solid phase, and $L$ is the latent heat per unit volume. $W_a=3\sigma_c / \xi_c$ represents an energy hump \cite{boettinger2002phase} with $\sigma_c$ denoting the surface tension coefficient. The second term on the right-hand side of Eq.~(\ref{fc}) describes the influence of latent heat of fusion for phase changes with an interpolating function $H_2(c)$, which is given as $H_2'(c)=30 H_1(c)$ with prime representing the derivative with respect to $c$ in Refs \cite{zhang2022phase, boettinger2002phase}, and $H_2 = 0$ in \cite{yang2006numerical, shen2012modeling, shen2010phase, shen2009efficient, huang2020consistent, liu2003phase} and other literature, which can be viewed as specific cases with the second term of the right-hand side of Eq.~(\ref{fc}) not considered.

The Cahn-Hilliard equation \cite{cahn1958free} is employed to model the evolution of the interface between water and air:
\begin{equation}
	\frac{\partial \phi}{\partial t} + \boldsymbol u \cdot \nabla\phi = M_\phi \nabla ^{2} \mu_{\phi} 
 \label{ch},
\end{equation}
where another phase-field order parameter $\phi$ is used with $\phi=1$ denoting the water phase and $\phi=-1$ the air phase. Similar to $c$, the definition of the range for $\phi$ is also arbitrary, and in this paper, we adopt the widely-used range of $[-1, 1]$. $M_\phi$ is the phenomenological mobility, and the chemical potential $\mu_{\phi}$ is \cite{zhang2022phase, xu2018sharp}:
\begin{equation}
	\mu_\phi= -\frac{3\sigma_\phi}{2\sqrt 2 \xi_\phi}\left [ \xi^2_\phi \nabla ^{2}\phi -G'(\phi)\right ],
 \label{muphi}
\end{equation}
where $\sigma_\phi$ is the surface tension coefficient, $\xi_\phi$ represents the thickness of the water-air interface, and $G'(\phi)$ denotes the derivative of $G(\phi)$ with respect to $\phi$, where $G(\phi) = \frac{1}{4} (\phi-1)^2 (\phi+1)^2 $ is the double-well function with bistable solutions at both $\phi = -1$ and $\phi = 1$.

The volume fractions of the air, water, and ice in the proposed model are defined as $\alpha_1=(1-\phi)/2, \alpha_2=(1+c)(1+\phi)/2, \alpha_3=(-c)(1+\phi)/2$, where the subscripts 1, 2, and 3 represent pure air, pure water, and pure ice respectively \cite{zhang2022phase, huang2022consistent}. $\alpha_1+\alpha_2+\alpha_3=1$ is always true in the whole domain.

Based on the volume fractions of the phases, density, viscosity, thermal conductivity, and specific heat capacity are obtained by linear interpolations \cite{zhang2022phase, huang2022consistent}:
\begin{equation} 
\begin{aligned}
	\rho(c,\phi)&=\alpha_1 \rho_1 + \alpha_2 \rho_2 + \alpha_3 \rho_3,\\
    \mu(c,\phi)&=\alpha_1 \mu_1 + \alpha_2 \mu_2 + \alpha_3 \mu_3, \\   \kappa(c,\phi)&=\alpha_1 \kappa_1 + \alpha_2 \kappa_2 + \alpha_3 \kappa_3,\\    c_p(c,\phi)&=\alpha_1 c_{p1} + \alpha_2 c_{p2} + \alpha_3 c_{p3},
 \label{properties}
\end{aligned}
\end{equation}
where $\rho$ is the density, $\mu$ is the viscosity, $\kappa$ is the thermal conductivity, and $c_p$ is the specific heat capacity. 

The motion of the flow is governed by the NS equations:
\begin{equation}
	\rho \frac{\partial \boldsymbol u}{\partial t} + \rho (\boldsymbol u \cdot \nabla) \boldsymbol u  =-\nabla p + \nabla\cdot \boldsymbol\tau + \boldsymbol F_b,
 \label{momentum}
\end{equation}
\begin{equation}
    \nabla \cdot \boldsymbol{u} =0,
    \label{mass}
\end{equation}
where $\boldsymbol\tau= \mu D(\boldsymbol u)+\boldsymbol\tau_e $ is the total stress tensor with $D(\boldsymbol u)=\nabla\boldsymbol u+(\nabla\boldsymbol u)^T$, superscript $T$ representing the transpose of a matrix, and $\boldsymbol \tau_e$ the extra elastic stress induced by the interfacial surface tension. $\boldsymbol F_b=\rho \boldsymbol g$ is the body force with $\boldsymbol g$ denoting the gravitational acceleration.

The extra elastic stress induced by the interfacial surface tension between different phases can be written as \cite{shen2012modeling, liu2003phase, liu2001eulerian}:
\begin{equation}
\boldsymbol\tau_e= -\lambda_c\nabla c \otimes \nabla c -\lambda_\phi\nabla\phi \otimes \nabla\phi,
\end{equation}
where $\lambda_c=\frac{3}{2\sqrt{2}}\sigma_c \xi_c$ and $\lambda_\phi=\frac{3}{2\sqrt{2}}\sigma_{\phi} \xi_\phi$ are the mixing energy density parameters of the ice-water interface and the water-air interface, respectively \cite{shen2012modeling}.
 
The divergence of $\boldsymbol\tau_e$ is:
\begin{equation}
\begin{aligned}
    \nabla\cdot\boldsymbol\tau_e =
    -\lambda_c(\nabla^{2}c\nabla c-\frac{1}{2} \abs{\nabla c}^2)
    -\lambda_\phi(\nabla^{2}\phi\nabla\phi-\frac{1}{2} \abs{\nabla\phi}^2).
    \label{taue}
\end{aligned}
\end{equation}

Absorbing the gradient terms in Eq.~(\ref{taue}) into the pressure $p$, the NS equations can be rewritten as:
\begin{equation}
\begin{aligned}
	\rho \frac{\partial \boldsymbol u}{\partial t} + \rho (\boldsymbol u \cdot \nabla) \boldsymbol u  =-\nabla p  + \nabla \cdot (\mu D(\boldsymbol u ))  -\lambda_c\nabla^{2}c\nabla c -\lambda_\phi\nabla^{2}\phi\nabla\phi + \boldsymbol F_b,
\label{momentum_2}
\end{aligned}
\end{equation}
\begin{equation}
    \nabla \cdot \boldsymbol{u} =0,
    \label{mass_2}
\end{equation}
where, for simplicity, $p$ is still used to denote the modified pressure.

The energy equation for icing is employed in this paper to guarantee the energy conservation of the system \cite{zhang2022phase}:
\begin{equation}
\begin{aligned}
\frac {\partial (\rho c_p T)}{\partial t}+\nabla \cdot (\rho c_p T \boldsymbol u) =\frac{\partial \alpha_3}{\partial t}L+\nabla\cdot(\kappa \nabla T )+\Phi + \boldsymbol u \cdot \boldsymbol F_b,
\label{energy}
\end{aligned}
\end{equation}
where the first term on the right side depicts the effect of the phase change to the enthalpy \cite{zhang2022phase, huang2022consistent}. $\Phi=(\mu D(\boldsymbol u))\boldsymbol:\nabla\boldsymbol u$ is the viscous dissipation rate, where $\boldsymbol:$ is the double inner product defined as $\boldsymbol a \boldsymbol :\boldsymbol b=a_{ij}b_{ji}$. $\rho$, $c_p$, $\mu$ and $\kappa$ are computed by Eq.~(\ref{properties}).

As shown above, the AC-CH-NS phase-field model for three-phase flows with phase change is established by combining the AC equation, CH equation, NS equations, and the energy equation, together with suitable boundary conditions, e.g.,
\begin{equation}
    \boldsymbol{u}|_{\partial \Omega} =0, \frac{\partial c}{\partial \boldsymbol n}|_{\partial \Omega} =0, \frac{\partial \phi}{\partial \boldsymbol n}|_{\partial \Omega} =0,\frac{\partial \mu_\phi}{\partial \boldsymbol n}|_{\partial \Omega} =0,\frac{\partial p}{\partial \boldsymbol n}|_{\partial \Omega} =0, T|_{\partial \Omega} = T_w,
    \label{bds}
\end{equation}
and initial conditions
\begin{equation}
    \boldsymbol{u}|_{t=0} =\boldsymbol{u}^0, c|_{t=0} =c^0, \phi|_{t=0} =\phi^0, \mu_\phi|_{t=0} =\mu^0_{\phi}, p|_{t=0} =p^0, T|_{t=0} =T^0, 
\end{equation}
where $\partial \Omega$ denotes the boundary of the domain, $\boldsymbol{n}$ is the unit normal vector pointing to the outer boundaries, $T_w$ is the temperature of the boundary, and superscript $0$ denotes the values at the initial time step. 

The proposed model satisfies the energy dissipation law, a basic property of most phase-field models \cite{shen2012modeling, yang2023modified2, yang2023phase2}. As described in Ref. \cite{shen2012modeling}, the terms $\nabla ^2 c \nabla c$ and $\nabla ^2 \phi \nabla \phi$ in Eq.~(\ref{momentum_2}) introduce difficulty in proving this law. Therefore, taking the advantage of Eq.~(\ref{ac}) and Eq.~(\ref{muphi}), we reformulate Eq.~(\ref{momentum_2}) by replacing $\nabla ^2 c \nabla c$ and $\nabla ^2 \phi \nabla \phi$ by $\frac{1}{M_c \xi^2_c} (\frac{\partial c}{\partial t} + \boldsymbol{u} \cdot \nabla  c) \nabla c$ and $-\frac{2\sqrt{2}}{3 \sigma_\phi \xi_\phi} \mu_\phi \nabla \phi$, respectively, after absorbing the terms $\frac{1}{\xi^2_c}F'(c)\nabla c = \frac{1}{\xi^2_c}\nabla (F(c))$ and $\frac{1}{\xi^2_\phi}G'(\phi)\nabla \phi= \frac{1}{\xi^2_\phi}\nabla (G(\phi))$ into $\nabla p$ (but still use the some notation for the modified $p$). The reformulated equivalent NS equations are:
\begin{equation}
\begin{aligned}
    \rho \frac{\partial \boldsymbol u}{\partial t} + \rho (\boldsymbol u \cdot \nabla) \boldsymbol u  = &-\nabla p + \nabla \cdot (\mu D(\boldsymbol u ))  \\
    &- \frac{\lambda_c }{M_c \xi^2_c} (\frac{\partial c}{\partial t} +  \boldsymbol{u} \cdot \nabla  c) \nabla c
    +\mu_\phi\nabla\phi
    + \boldsymbol F_b,
\label{momentum_incomp}
\end{aligned}
\end{equation}
\begin{equation}
    \nabla \cdot \boldsymbol{u} =0.
    \label{mass_3}
\end{equation}
\newtheorem{theorem}{Theorem}[section]

\begin{theorem}
The solutions of the Eqs.~(\ref{ac}), (\ref{ch}), (\ref{muphi}), (\ref{momentum_incomp}), and (\ref{mass_3}) with boundary conditions (\ref{bds}) dissipate a total energy functional as follows.
\begin{equation}
    E_0 = \int_\Omega  
    (\frac{1}{2} B_1 \xi^2_c \abs{\nabla c}^2
    + B_1 F(c)
    + \frac{1}{2} B_2 \xi^2_\phi \abs{\nabla \phi}^2
    + B_2 G(\phi)
    + \frac{\rho}{2} \abs{\boldsymbol{u}}^2
    + E_g)
    d \boldsymbol {x},
    \label{E0}
\end{equation}
where $B_1=\frac{3\sigma_c}{2\sqrt{2}\xi_c}$, and $B_2 = \frac{3\sigma_\phi}{2\sqrt{2}\xi_\phi}$ are positive constants. $E_g = \rho g h$ is the gravitational potential energy with $h$ the height.
\end{theorem}

\begin{proof}
Multiplying Eq.~(\ref{ac}) with $\frac{\partial c}{\partial t}$ and taking the $L^2$-inner product yields:
\begin{equation}
 \norm{\frac{\partial c}{\partial t}}^2 + (\boldsymbol{u} \cdot \nabla  c, \frac{\partial c}{\partial t} ) = - \frac{d}{dt}\int _\Omega \frac{1}{2} M_c \xi^2_c \abs{\nabla c}^2 d \boldsymbol {x} - \frac{d}{dt}\int _\Omega M_c F(c) d \boldsymbol {x},
    \label{ac_l2}
\end{equation}
where the $L^2$-inner product is defined as $(g_a, g_b)=  \int _\Omega g_a\cdot g_b\ d \boldsymbol {x}$ with $g_a$ and $g_b$ denoting two integrable functions, the $L^2$-norm is defined as $\norm {g_a}^2=(g_a, g_a)$, and the integral by parts is used \cite{shen2012modeling}.

Multiplying Eq.~(\ref{ch}) with $\mu_\phi$ and taking the $L^2$-inner product results in:
\begin{equation}
    (\frac{\partial \phi}{\partial t}, \mu_\phi) + (\boldsymbol{u}\cdot \nabla \phi, \mu_\phi) = -M_\phi \norm {\nabla \mu_\phi}^2.
    \label{ch_l2}
\end{equation}

By multiplying Eq.~(\ref{muphi}) with $\frac{\partial \phi}{\partial t}$ and taking the $L^2$-inner product, we get:
\begin{equation}
    (\mu_\phi, \frac{\partial \phi}{\partial t})= \frac{d}{dt} \int_\Omega B_2 G(\phi)\ d \boldsymbol {x} + \frac{d}{dt} \int_\Omega \frac{1}{2} B_2 \xi^2_\phi \abs{\nabla \phi}^2 d \boldsymbol {x}.
    \label{muphi_l2}
\end{equation}

Taking the $L^2$-inner product of Eq.~(\ref{momentum_incomp}) with $\boldsymbol{u}$ results in:
\begin{equation}
\begin{aligned}
       \frac{d}{dt} \int_\Omega \frac{\rho}{2} \abs{\boldsymbol{u}}^2 d \boldsymbol {x} =
    &  - \norm {\sqrt{\mu} D(\boldsymbol{u})}^2
       - \frac{B_1}{M_c}( (\frac{\partial c}{\partial t}\nabla c,\boldsymbol{u}) + \norm {\boldsymbol{u}\cdot \nabla c}^2 )\\
    &  + (\mu_\phi \nabla \phi, \boldsymbol {u})
    - \frac{d}{dt} \int_\Omega E_g d \boldsymbol {x},
    \label{ns_l2}
\end{aligned}
\end{equation}
where ($\boldsymbol{u} \cdot \nabla\boldsymbol{u}$, $\boldsymbol{u}$) $=$ 0 and ($-\nabla p, \boldsymbol{u}$) =0 are used \cite{yang2023modified2}.

By combining Eq.~(\ref{ns_l2}) with $- \frac{B_1}{M_c}$ timing Eq.~(\ref{ac_l2}), Eq.~(\ref{ch_l2}), and Eq.~(\ref{muphi_l2}), we obtain:
\begin{equation}
\begin{aligned}
         \frac{d}{dt} E_0 = - \norm {\sqrt{\mu}D(\boldsymbol{u})}^2         
         - \frac{B_1}{M_c}\norm {\frac{\partial c}{\partial t}+ \boldsymbol{u}\cdot \nabla c}^2  
         -M_\phi \norm {\nabla \mu_\phi}^2\leq 0.
\end{aligned}
\end{equation}

The above inequality completes the proof.
\end{proof}

\section{Computational algorithm} \label{3}
\subsection{Temporal discretization scheme}
In this section, we present an energy-stable first-order temporal discretization scheme for the established model based on the pressure stabilization method \cite{guermond2009splitting}.

For the temporal discretization, $t_0$ represents the total computational time, divided into $N$ steps with time step $\Delta t = t_0/N$. The variables at time step $n$ are defined as $(\cdot)^n$. Given the initial conditions of the variables $\boldsymbol{u}^0, c^0, \phi^0, \mu^0_\phi, p^0, T^0$, the variables at time step $n+1$ are calculated as follows.

The temporal discretization for the AC equation (\ref{ac}) is:
\begin{equation}
    \frac{c^{n+1}-c^n}{\Delta t} + S_1(c^{n+1}-c^n) + \boldsymbol{u}^{n+1}\cdot \nabla c^n = M_c[\xi_c ^2 \nabla^2c^{n+1}-F'(c^n)],
    \label{ac_time}
\end{equation}
where $S_1$ is a stabilizing parameter \cite{shen2012modeling} for the AC equation.

The CH equations (\ref{ch})-(\ref{muphi}) are temporally discretized as:
\begin{equation}
    \frac{\phi^{n+1}-\phi^n}{\Delta t} + \boldsymbol{u}^{n+1}\cdot \nabla \phi^n = M_\phi \nabla^2 \mu_{\phi} ^{n+1},
    \label{ch_time}
\end{equation}

\begin{equation}
	\mu_\phi ^{n+1} -S_2(\phi^{n+1}-\phi^n)= -\frac{3\sigma_\phi}{2\sqrt 2 \xi_\phi}\left [ \xi^2_\phi \nabla ^{2}\phi^{n+1} -G'(\phi^n)\right],
 \label{muphi_time}
\end{equation}
where $S_2$ is the stabilizing parameter for the CH equation \cite{shen2012modeling}.

$\frac{1}{\Delta t}[\frac{1}{2}(\rho^{n+1}+\rho^{n})\boldsymbol{u}^{n+1}-\rho^n \boldsymbol{u}^{n}] + \frac{1}{2} \nabla \cdot (\rho^n \boldsymbol{u}^n) \boldsymbol{u}^{n+1}$ is applied for the discretization of the first term $\rho \frac{\partial \boldsymbol{u}}{\partial t}$ in the NS equations (\ref{momentum_incomp}) \cite{guermond2009splitting}. Therefore, the pressure stabilization schemes for the NS equations are \cite{guermond2009splitting}:
\begin{equation}
    \begin{aligned}
       \frac{1}{\Delta t}[\frac{1}{2}(\rho^{n+1}& +\rho^{n})\boldsymbol{u}^{n+1}-\rho^n \boldsymbol{u}^{n}] + \frac{1}{2} \nabla \cdot (\rho^n \boldsymbol{u}^n) \boldsymbol{u}^{n+1} + \rho^n (\boldsymbol{u}^n \cdot \nabla) \boldsymbol{u}^{n+1}\\
      & = 
    - \nabla (2 p^n- p^{n-1})
    + \nabla \cdot (\mu^{n+1} D(\boldsymbol{u}^{n+1}))\\
        &- \frac{\lambda_c}{M_c \xi^2_c} (\frac{c^{n+1}-c^n}{\Delta t} + \boldsymbol{u}^{n+1} \cdot \nabla c^n) \nabla c^n 
        + \mu_\phi ^{n+1} \nabla \phi ^n
        + \boldsymbol{F}_b ^n,
        \label{NS1_time}
    \end{aligned}
\end{equation}
\begin{equation}
    \nabla^2(p^{n+1}-p^n) = \frac{\bar{\rho}}{\Delta t} \nabla \cdot \boldsymbol{u}^{n+1},
    \label{NS2_time}
\end{equation}
where $\rho^{n+1} =\alpha^{n+1}_1 \rho_1 + \alpha^{n+1}_2 \rho_2 + \alpha^{n+1}_3 \rho_3, \mu^{n+1}=\alpha^{n+1}_1 \mu_1 + \alpha^{n+1}_2 \mu_2 + \alpha^{n+1}_3 \mu_3$ with $\alpha^{n+1}_1=(1-\phi^{n+1})/2, \alpha^{n+1}_2=(1+c^{n+1})(1+\phi^{n+1})/2$, and $\alpha^{n+1}_3=(-c^{n+1})(1+\phi^{n+1})/2$. $\bar{\rho}= {\rm min}(\rho_1,\rho_2,\rho_3)$ is the minimum density of three phases.

The temporal discretization for the energy equation (\ref{energy}) is:
\begin{equation}
\begin{aligned}
&\frac { (\rho c_p T)^{n+1}- (\rho c_p T)^{n}}{\Delta t} +\nabla \cdot (\rho c_p T \boldsymbol u)^{n+1} \\
&=\frac{\alpha^{n+1}_3-\alpha^{n}_3}{\Delta t}L+\nabla\cdot(\kappa^{n+1} \nabla T^{n+1} )+\Phi^{n+1} + \boldsymbol u^{n+1} \cdot \boldsymbol F^{n+1}_b,
\label{energy_time}
\end{aligned}
\end{equation}
where $c_p^{n+1} =\alpha^{n+1}_1 c_{p_1} + \alpha^{n+1}_2 c_{p_2} + \alpha^{n+1}_3 c_{p_3}, \kappa^{n+1}=\alpha^{n+1}_1 \kappa_1 + \alpha^{n+1}_2 \kappa_2 + \alpha^{n+1}_3 \kappa_3$.

The boundary conditions at time step $n+1$ are as follows:
\begin{equation}
\begin{aligned}
   & \boldsymbol{u}^{n+1}|_{\partial \Omega} =0,
     \frac{\partial }{\partial \boldsymbol n}c^{n+1}|_{\partial \Omega} =0, \frac{\partial }{\partial \boldsymbol n}\phi^{n+1}|_{\partial \Omega} =0,\\
   & \frac{\partial }{\partial \boldsymbol n}\mu^{n+1}_\phi|_{\partial \Omega} =0,
   \frac{\partial }{\partial \boldsymbol n}p^{n+1}|_{\partial \Omega} =0, T^{n+1}|_{\partial \Omega} = T_w.
    \label{bds_times}
\end{aligned}
\end{equation}

\begin{theorem} 
\label{theorem 3}
For $S_1 \geq \frac{1}{2} M_c$ and $S_2 \geq B_2$, the scheme including Eqs.~(\ref{ac_time})-(\ref{bds_times}) satisfies the following discrete energy law:
\begin{equation}
\begin{aligned}
&\norm{\sqrt{\rho^{n+1}}\boldsymbol{u}^{n+1}}^2
+ \frac{(\Delta t)^2}{\bar{\rho}} \norm{\nabla p^{n+1}}^2
+ B_1 \xi^2_c \norm {\nabla c^{n+1}} ^2 
+2 B_1 (F(c^{n+1}), 1)
\\
&+ B_2 \xi^2_\phi \norm {\nabla \phi^{n+1}}^2
+2 B_2(G(\phi^{n+1}), 1) 
+2 (E^{n+1}_g, 1)\\
&+ 2 \Delta t (        
   \norm{\sqrt{\mu^{n+1}}D(\boldsymbol{u}^{n+1})}^2
+\frac{B_1}{M_c} \norm{\frac{c^{n+1}-c^n}{\Delta t} + \boldsymbol{u}^{n+1} \cdot \nabla c^n }^2
+ M_\phi \norm {\nabla \mu_\phi^{n+1}}^2) \\
& \leq 
\norm{\sqrt{\rho^{n}}\boldsymbol{u}^{n}} ^2
+ \frac{(\Delta t)^2}{\bar{\rho}} \norm{\nabla p^n}^2
+  B_1 \xi^2_c  \norm {\nabla c^{n}}^2
+2 B_1 (F(c^{n}), 1) 
\\
&+ B_2 \xi^2_\phi \norm {\nabla \phi^{n}}^2
+2 B_2 (G(\phi^{n}), 1)
+2 (E^{n}_g, 1).
\label{discrete_energy_law}
\end{aligned}
\end{equation}
\end{theorem}

The proof of Theorem~\ref{theorem 3} is presented in \ref{app}. In this work, $S_1 = \frac{1}{2} M_c$ and $S_2 = B_2$ are applied. This theorem demonstrates that the discrete energy at time step $n+1$ is smaller than that at time step $n$, and therefore the scheme is unconditionally stable with the discrete energy at time step $n$ defined as the right-hand side of inequality (\ref{discrete_energy_law}).

\subsection{Spatial discretization scheme}
The differential operators are discretized by the finite difference method (FDM) with second-order central difference spatial discretization applied on the staggered grids. The main advantage of staggered grids over collocated grids is that they prevent the decoupling of pressure and velocity, thereby avoiding the checkerboard problem \cite{kawaguchi2002checkerboard,khorrami1991chebyshev}. The scalars such as phase-field order parameter $\phi$, $c$, density $\rho$, viscosity $\mu$, thermal conductivity $\kappa$, heat capacity $c_p$, pressure $p$, and temperature $T$ are stored at the cell centres, while vectors such as the velocities are evaluated at the cell surfaces. A two-dimensional formulation in Cartesian coordinates is presented as an example, which can be extended to three dimensions or degraded to one dimension straightforwardly. In this way, the cell centre is denoted as ($x_i, y_j$) while the cell surfaces are denoted as ($x_{i-\frac{1}{2}},y_j$), ($x_{i+\frac{1}{2}},y_j$), ($x_i,y_{j-\frac{1}{2}}$), and ($x_i,y_{j+\frac{1}{2}}$). As shown in Fig.~\ref{staggered grids}, the vector value $\boldsymbol F$ defined at the cell surfaces is written as $\boldsymbol F^x_{i-\frac{1}{2},j}$, $\boldsymbol F^x_{i+\frac{1}{2},j}$, $\boldsymbol F^y_{i,j-\frac{1}{2}}$, and $\boldsymbol F^y_{i,j+\frac{1}{2}}$ with subscripts denoting the positions of grids and superscripts representing $x-$ or $y-$ components of the vector. While the scalar parameters at the cell centres are marked as $\mathscr{F}_{i,j}$ or $\mathscr{H}_{i,j}$ in a similar manner. The spatial discretization applied for the central-difference interpolation, gradient terms, divergence terms, and Laplace terms used in the governing equations is presented as follows.
\begin{figure}[ht]
    \centering 
    \includegraphics[width=0.8\textwidth]{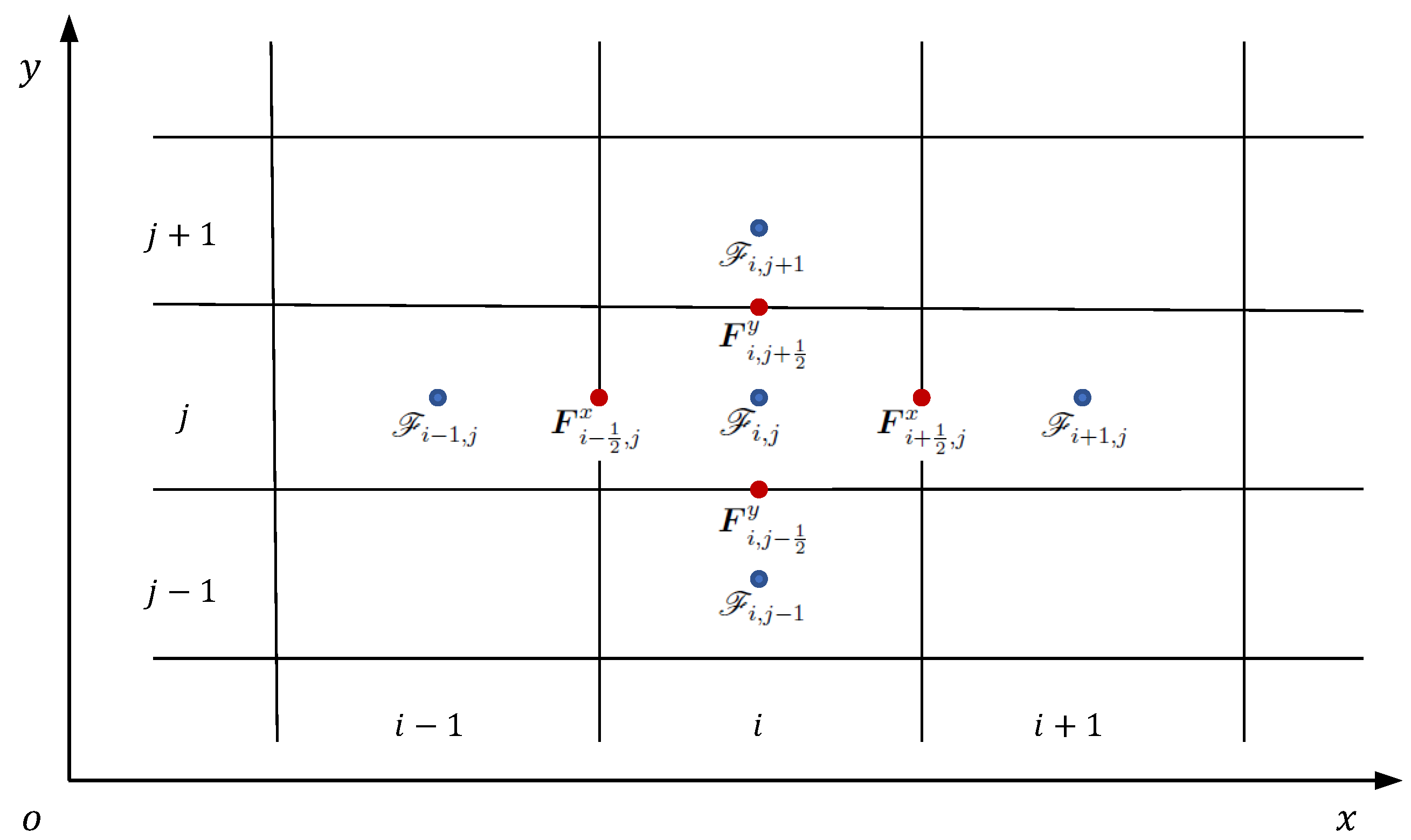}
    \caption{Staggered grids discretization.} %
\label{staggered grids}
\end{figure}

i) Central-difference interpolation

Central-difference interpolation is applied for the calculation of vectors at the cell centres and scalars at cell surfaces, e.g., $\boldsymbol{F}^x_{i,j} = ( \boldsymbol{F}^x_{i+\frac{1}{2},j}+\boldsymbol{F}^x_{i-\frac{1}{2},j})/2$.  $\mathscr{F}_{i,j+\frac{1}{2}} =(\mathscr{F}_{i,j+1}+\mathscr{F}_{i,j})/2$.

ii) Discrete gradient

The discrete gradient of a scalar at the cell centre is calculated as:
\begin{equation}
     \nabla \mathscr{F}_{i,j}
    = \frac{\mathscr{F}_{i+1,j}-\mathscr{F}_{i-1,j}}{2\Delta x} \boldsymbol i_x +\frac{\mathscr{F}_{i,j+1}-\mathscr{F}_{i,j-1}}{2\Delta y} \boldsymbol i_{y},
\end{equation}
where $\Delta x$ and $\Delta y$ stand for the cell size along the $x$ and $y$ directions, $\boldsymbol i_x$ and $\boldsymbol i_{y}$ are the standard unit vectors for the $x-$ and $y-$axes.

The convection term in Eq.~(\ref{ac}) or Eq.~(\ref{ch}), evaluated at the cell centre, is thus computed as:
\begin{equation}
\begin{aligned}
    \left (\boldsymbol F \cdot \nabla \mathscr{F} \right )_{i,j}
    =& \frac{(\boldsymbol F^x_{i+\frac{1}{2},j} + \boldsymbol F^x_{i-\frac{1}{2},j})(\mathscr{F}_{i+1,j}-\mathscr{F}_{i-1,j})}{4\Delta x}+ \\
& \frac{(\boldsymbol F^{y}_{i,j+\frac{1}{2}} + \boldsymbol F^{y}_{i,j-\frac{1}{2}})(\mathscr{F}_{i,j+1}-\mathscr{F}_{i,j-1})}{4\Delta y}.   
\end{aligned}
\end{equation}

iii) Discrete divergence

The divergence evaluated at the cell centre exhibits two different forms in the convection and diffusion terms in Eq.~(\ref{energy}). For the convection term in Eq.~(\ref{energy}), the divergence is 
\begin{equation}
\begin{aligned}
    \left [\nabla \cdot (\boldsymbol F \mathscr{F}) \right ]_{i,j}
    =&\frac{ \boldsymbol F^x_{i+\frac{1}{2},j} (\mathscr{F}_{i+1,j} + \mathscr{F}_{i,j}) -  \boldsymbol F^x_{i-\frac{1}{2},j} (\mathscr{F}_{i-1,j} + \mathscr{F}_{i,j})}{2 \Delta x}\\
    &+\frac{ \boldsymbol F^{y}_{i,j+\frac{1}{2}} (\mathscr{F}_{i,j+1} + \mathscr{F}_{i,j}) -  \boldsymbol F^{y}_{i,j-\frac{1}{2}} (\mathscr{F}_{i,j-1} + \mathscr{F}_{i,j})}{2 \Delta y},
\end{aligned}
\end{equation}
and the divergence for the diffusion term in the energy equation (\ref{energy}) is
\begin{equation}
\begin{aligned}
    &\left [\nabla \cdot (\mathscr{F} \nabla \mathscr{H}) \right ]_{i,j}=\\
   & \frac{ (\mathscr F_{i+1,j}+ \mathscr F_{i,j})(\mathscr H_{i+1,j}-\mathscr H_{i,j}) - (\mathscr F_{i-1,j}+ \mathscr F_{i,j})(\mathscr H_{i,j}-\mathscr H_{i-1,j})}{2(\Delta x)^2}\\
    & + \frac{ (\mathscr F_{i,j+1}+ \mathscr F_{i,j})(\mathscr H_{i,j+1}-\mathscr H_{i,j}) -  (\mathscr F_{i,j-1}+ \mathscr F_{i,j})(\mathscr H_{i,j}-\mathscr H_{i,j-1})}{2(\Delta y)^2}.
\end{aligned}
\end{equation}

iv) Discrete Laplace operator

The discrete Laplace operator at the cell centre is calculated as:
\begin{equation}
    \begin{aligned}
        \nabla^2 \mathscr F_{i,j} = \frac{\mathscr F_{i+1,j} -2 \mathscr F_{i,j} +\mathscr F_{i-1,j} }{(\Delta x)^2} + \frac{\mathscr F_{i,j+1} - 2\mathscr F_{i,j}+ \mathscr F_{i,j-1}}{(\Delta y)^2}.
    \end{aligned}
\end{equation}

\section{Numerical results} \label{4}
\begin{table*}[htbp]
	\centering
	\caption{Properties of air, water and ice}
	\begin{tabular}{ccccc}
		\toprule  
		Phase & \makecell[c]{Density\\ ($kg \cdot m^{-3}$)}  & \makecell[c]{Viscosity\\ ($Pa \cdot s$)} & \makecell[c]{Thermal conductivity\\ ($W \cdot m^{-1} \cdot K^{-1})$}  &\makecell[c]{Heat capacity\\ ($J \cdot K^{-1} \cdot kg^{-1}$)}   \\ 
		\midrule  
		air&1.2&$1.6\times10^{-5}$&0.0209&1003 \\
		\midrule
		water&998&$1\times10^{-3}$&0.5918&4200 \\
		\midrule
		ice&898&$1\times10^{2}$&2.25&2018 \\
		\bottomrule  
	\end{tabular} 
 \label {table properties}
\end{table*}

In this section, the droplet icing process on supercooled substrates is numerically simulated using the phase-field model with three phases established in Section \ref{3}. The time evolution of the droplet profiles and temperature fields, the influence of the supercooled substrates, and the effects of the surrounding air are studied, and the evolution of the energy $E_0$ is plotted, numerically demonstrating that the proposed scheme is energy stable. 
\begin{table*}[htbp]
	\centering
	\caption{Values of physical parameters}
	\begin{tabular}{cccc}
		\toprule  
		Parameter & Symbol & Value  & Unit    \\ 
		\midrule  
		Mobility of ice-water interface & $M_c$ & $2.00\times10^{-3}$ & $N^{-1} \cdot s^{-1} \cdot m^{2}$ \\
		\midrule
	\makecell[c]{Surface tension coefficient \\ of ice-water interface} & $\sigma_c$ & $7.27\times10^{-2}$ & $N \cdot m^{-1}$ \\
		\midrule
	\makecell[c]{Thickness coefficient \\ of ice-water interface} & $\xi_c$ & $8.00\times10^{-5}$ & $m$ \\
            \midrule
        Mobility of water-air interface  &  $M_\phi$ & $2.50\times10^{-11}$ & $N^{-1} \cdot s^{-1} \cdot m^{4}$ \\
		\midrule
	\makecell[c]{Surface tension coefficient \\ of water-air interface}  & $\sigma_\phi$ & $7.27\times10^{-2}$ & $N \cdot m^{-1}$ \\
		\midrule
\makecell[c]{Thickness coefficient \\ of water-air interface} &	$\xi_\phi$ & $8.00\times10^{-5}$ & $m$ \\
		\midrule
Latent heat per unit volume & $L$ & $3.34\times10^{8}$ & $J \cdot m^{-3}$ \\
		\midrule
Melting temperature & $T_M$ & $2.73\times10^{2}$ & $K$ \\
		\bottomrule 
	\end{tabular} 
 \label {table parameters values}
\end{table*} 
The material properties of pure ice, water, and air are shown in Table \ref{table properties} and are assumed to be constant and independent of temperature and pressure. We note that the viscosity of ice is set as $10^5$ times that of water, which is large enough to suppress the movement of pure ice. The values of the other physical parameters shown in Table \ref{table parameters values} will be adopted if not specified \cite{zhang2022phase,carlson2009modeling,hillig1998measurement,jacqmin2000contact, wang2024consistent}. The dynamics of the droplet icing process are governed by the Prandtl number $Pr$, the Bond number $Bo$, the Weber number $We$, and the Stefan number $St$ \cite{vu2015numerical}:
\begin{equation}
    Pr=\frac{c_{p2}\mu_2}{\kappa_2}, Bo=\frac{\rho_2gR^2}{\sigma_c}, We=\frac{\rho_2U_c^2R}{\sigma_c}, St=\frac{\left(T_M-T_{sub}\right)c_{p2}\rho_2}{L},
\end{equation}
where $U_c=R/\tau_c$ is the velocity scale with $R$ the radius of the droplet and $\tau_c=\frac{\rho_2c_{p2}R^2}{\kappa_2}$ the characteristic time scale.

\subsection{Droplet icing process} \label{3.1}
Fig.~\ref{icing} shows typical snapshots during the freezing process of a droplet. The droplet is deposited on a supercooled substrate, confined between two parallel plates in a Hele-Shaw cell \cite{schremb2016solidification, tropea2017physics}, where it is squeezed into a thin, circular shape. The initial shape of the water droplet is approximated as a circular segment with a radius of 1.6 mm and a height of 2.28 mm. The calculation domain size is 8 mm $\times$ 8 mm $\times$ 0.08 mm with the gap between the plates being 0.08 mm, which corresponds to the thickness of the droplet. The substrate temperature is -25°C and the initial temperatures of both water and air are set as 25°C. The dimensionless numbers are $Pr$ = 7.1, $Bo$ = 0.35, $We$ = 1.7 $\times{10}^{-7}$, and $St$ = 0.31. 200 $\times$ 200 $\times$ 2 grids are used for the simulation, with a time step of 1.6 $\times$ 10$^{-6}$ seconds. The solidification process takes 2.4 $\times$ 10$^{7}$ steps, and the simulation is performed in parallel on 50 cores (CPU: Hygon G86 7285), running for 48.8 hours to obtain the numerical results. Red, peachy orange, and blue regions represent water, ice, and air, respectively.
\begin{figure}[ht!]
    \centering 
    \includegraphics[width=0.99\textwidth]{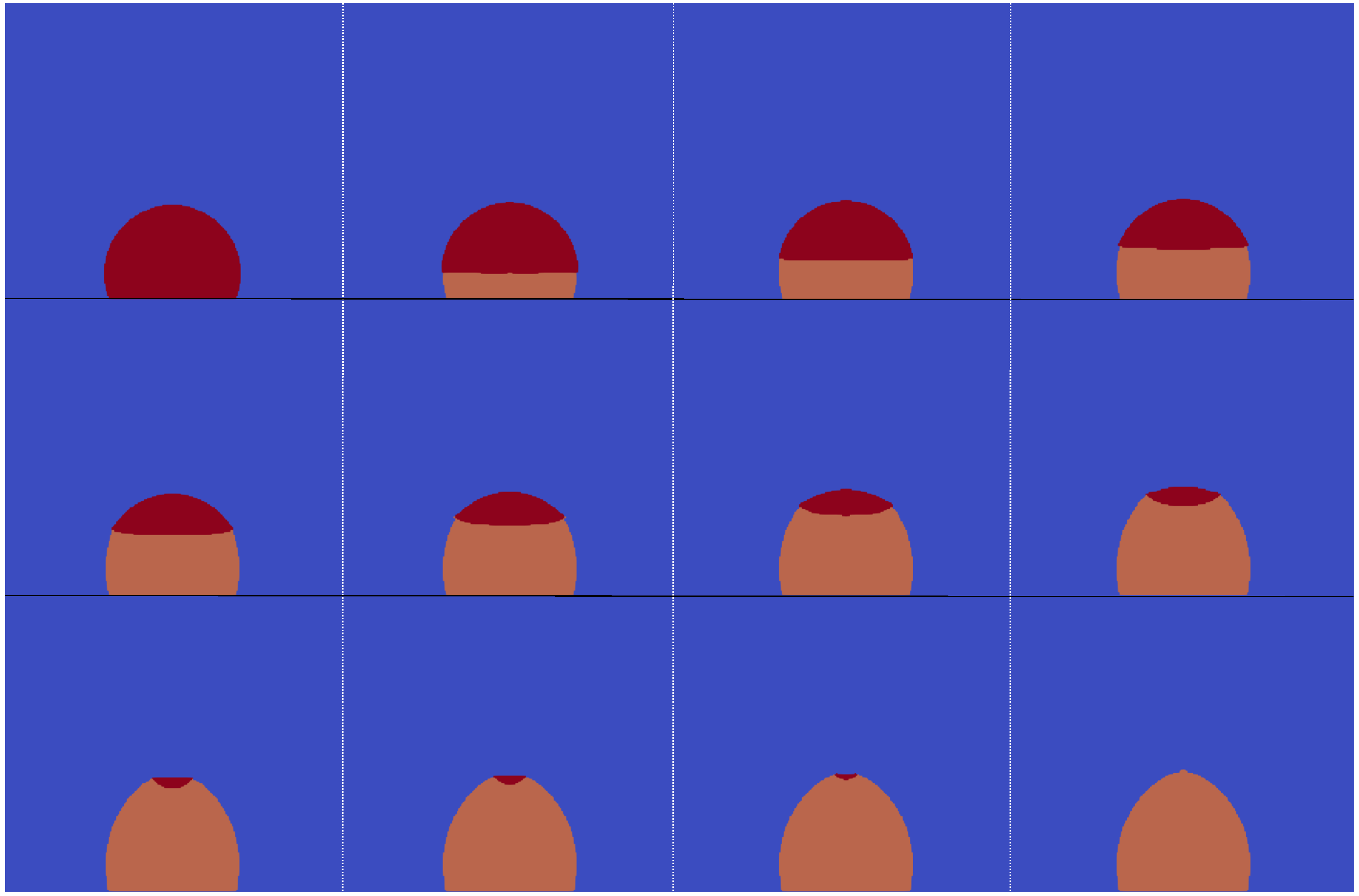}
    \caption{Snapshots of droplet icing. Red, peachy orange, and blue represent the water, ice, and air, respectively. The initial shape of the water droplet is approximated as a circular segment with a radius of 1.6 mm and a height of 2.28 mm. The substrate temperature is -25°C and the initial temperatures of both water and air are 25°C. From left to right and top to bottom, the times are t = 0, 4.6, 9.2, 13.8, 18.4, 23.0, 27.6, 32.2, 36.8, 37.3, 37.8, and 38.4 s.} %
\label{icing}
\end{figure}
The cooling effect from the cold substrate causes water droplets to freeze from the bottom up gradually. After approximately 38.4 seconds, the water droplet completely transforms into an ice droplet, with a pointy tip formed. During the solidification process, water droplets gradually increase in height due to expansion. Meanwhile, the width of the water droplets remains approximately constant. This is because the surface of the water droplets is cooled by the air, forming a thin layer of ice prematurely. This layer of ice prevents the water from expanding sideways, resulting in the primary expansion of the water droplets being upward during the freezing process. This upward expansion eventually “squeezes” out a tip (see the last snapshot of Fig.~\ref{icing}).

\begin{figure}[ht!]
    \centering 
    \includegraphics[width=0.99\textwidth]{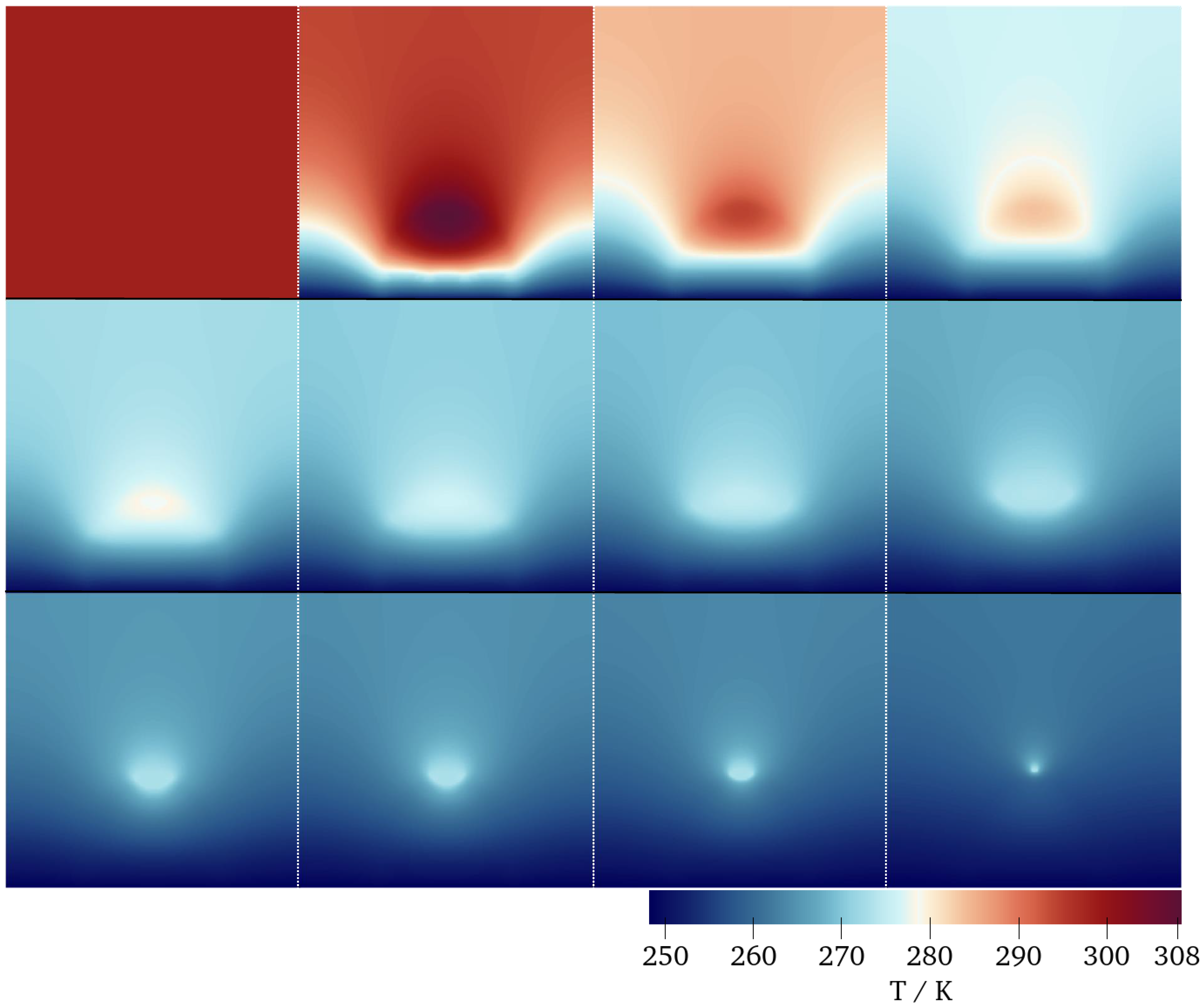}
    \caption{Snapshots of temperature field evolution during droplet icing process. The substrate temperature is -25°C (see the bottom of each image), and the initial temperatures of both water and air are 25°C. From left to right and top to bottom, the times are t = 0, 4.6, 9.2, 13.8, 18.4, 23.0, 27.6, 32.2, 36.8, 37.3, 37.8, and 38.4 s.} %
\label{icing_T}
\end{figure}
The corresponding evolution of the temperature field is shown in Fig.~\ref{icing_T}. The initial temperature of the water droplet as well as the ambient air is 25°C and the whole domain is cooled to -25°C by the cold substrate. The profiles of the positive temperature areas in the last 8 snapshots are consistent with the liquid regions of the droplet profiles in Fig.~\ref{icing}, demonstrating that the droplet freezing is mainly influenced by the heat transfer process. 

The air region on the two sides reaches low temperatures more rapidly than the water droplet. This is because air enjoys a larger thermal diffusivity $\alpha$, defined as $\frac{\kappa}{\rho c_p}$. This in turn proves the conclusion that the surface of the droplet solidifies more quickly than the inside resulting in the volume expansion mainly in the vertical direction.

\subsection{Effects of the substrate temperature}

\begin{figure}[ht!]
\centering
  \begin{minipage}[t]{0.3\linewidth}
    \centering 
    \includegraphics[width=1.0\textwidth]{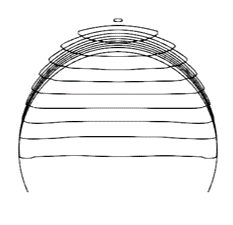}
    \caption*{(a)} 
   \end{minipage}
  \begin{minipage}[t]{0.3\linewidth}
    \centering 
    \includegraphics[width=1.0\textwidth]{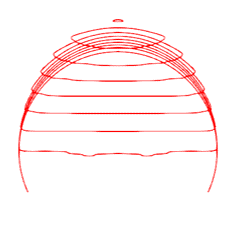}
    \caption*{(b)} 
   \end{minipage}
  \begin{minipage}[t]{0.3\linewidth}
    \centering 
    \includegraphics[width=1.0\textwidth]{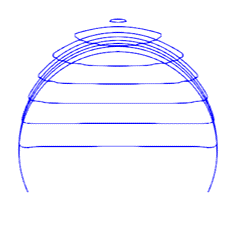}
    \caption*{(c)} 
   \end{minipage}
\caption{Droplet profiles and ice front evolutions of different substrate temperatures: (a) $T_{sub}$ = -20°C, (b) $T_{sub}$ = -25°C, and (c) $T_{sub}$ = -30°C. The initial shape of the water droplet is approximated as a circular segment with a radius of 1.6 mm and a height of 2.28 mm. From the bottom to the top, the nearly horizontal lines represent the evolution of the ice front, with a time interval of 4.6 s between adjacent lines and the first line corresponding to t = 4.6 s.} 
\label{icing_profile}
\end{figure}

\begin{figure}[ht!]
\centering
  \begin{minipage}[t]{0.49\linewidth}
    \centering 
    \includegraphics[width=1.0\textwidth]{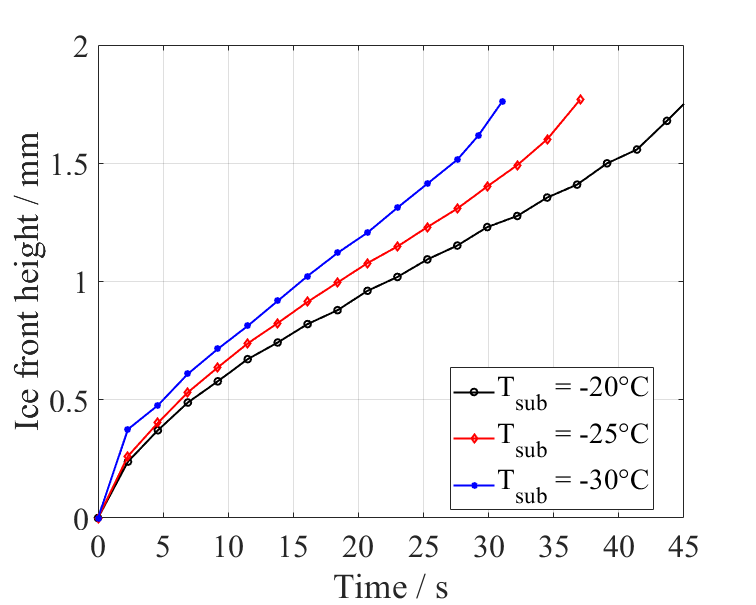}
    \caption*{(a)} 
\end{minipage}
\begin{minipage}[t]{0.49\linewidth}
    \centering 
\includegraphics[width=1.0\textwidth]{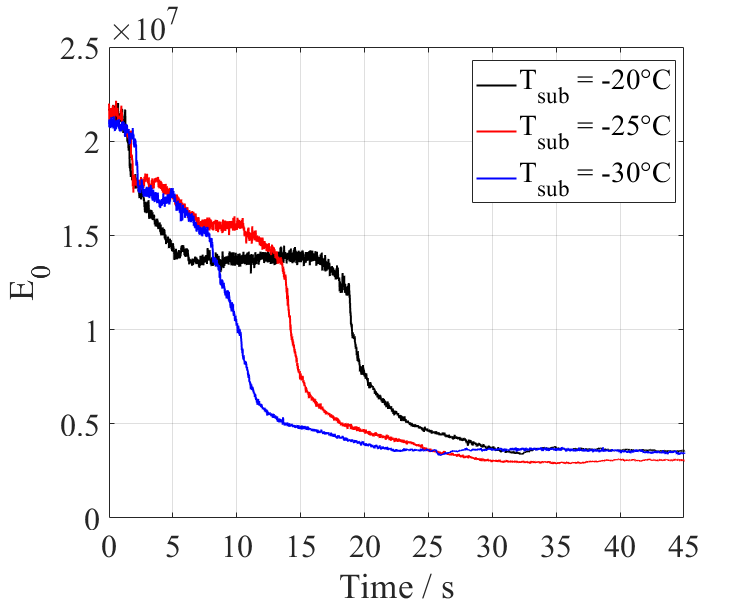}
    \caption*{(b)} 
\end{minipage}
\caption{Evolution of (a) ice front and (b) energy $E_0$ of different substrate temperatures. The radius of the droplets is 1.6 mm.} 
\label{icing_front_energy}
\end{figure}
\begin{figure}[ht!]
    \centering 
    \includegraphics[width=0.99\textwidth]{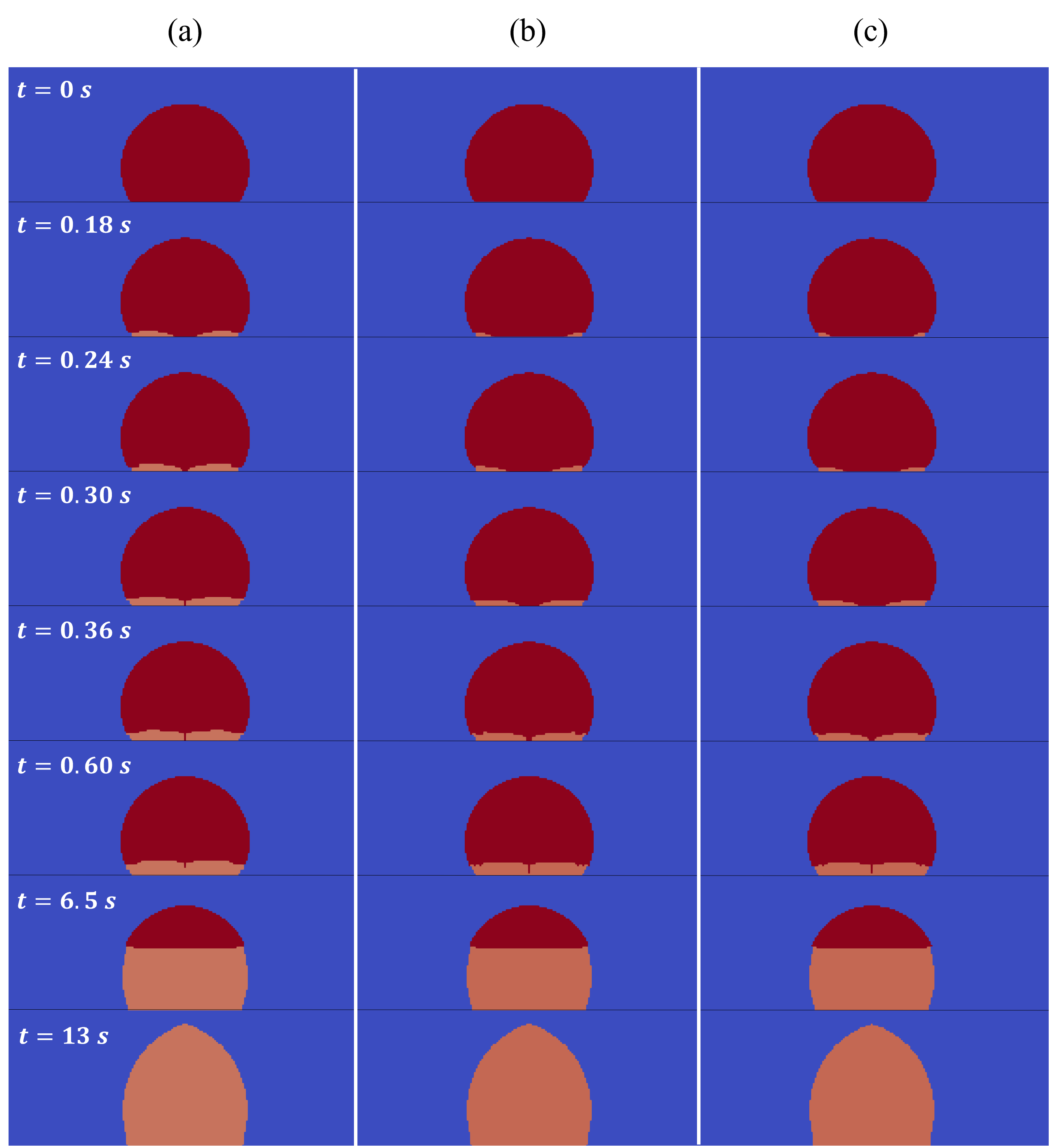}
    \caption{Ice front evolution with different air temperatures. From left to right, the columns represent (a) $T_a = -15 ^{\circ}$C, (b) $T_a = 5 ^{\circ}$C, and (c) $T_a = 25 ^{\circ}$C, respectively. The initial shape of the water droplet is approximated as a circular segment with a radius of 0.8 mm and a height of 1.36 mm. The temperature of substrates is $-25 ^{\circ}$C.} 
\label{Ta}
\end{figure}

Fig.~\ref{icing_profile}(a), (b) and (c) demonstrate the freezing front evolution driven by substrate temperature $T_{sub}$ of $-20 ^{\circ}$C, $-25 ^{\circ}$C, and $-30 ^{\circ}$C, respectively. The initial temperatures of the water droplets and their surrounding air are $25 ^{\circ}$C. The initial shape and radius of the droplets are the same as those in Section \ref{3.1}. The dimensionless numbers are $Pr$ = 7.1, $Bo$ = 0.35, $We$ = 1.7 $\times{10}^{-7}$, $St$ = 0.25, 0.31, and 0.38 for Fig.~\ref{icing_profile}(a), (b) and (c), respectively.

The lower temperature of the substrate speeds up the freezing process, as depicted in Fig.~\ref{icing_profile} and Fig.~\ref{icing_front_energy}(a), while leaving the final shape of the ice droplets unaffected. Due to the cooling effect of the air, the region of the droplet near the air reaches a lower temperature first, causing the freezing fronts to change from approximately a horizontal shape to a concave one, agreeing well with the experiments \cite{marin2014universality}. 

Fig.~\ref{icing_front_energy}(b) plots the time evolutions of the energy $E_0$ (see the formulaic expression in Eq.~(\ref{E0})). It can be seen that the energy for all the cases decreases with time before reaching equilibrium states. The computational results demonstrate that the proposed model satisfies the energy dissipation law.

\subsection{Effects of the air temperature}
Fig.~\ref{Ta} is the freezing front evolution driven by supercooled substrates with a temperature of -25°C. The domain is 8 mm $\times$ 8 mm $\times$ 0.08 mm with the droplet radius 0.8 mm (red region) surrounded by the air (blue region). The initial height of the droplet is set as 1.36 mm. Based on the established three-phase numerical model, the influence of the air is investigated by parameter studies. As shown in the figure, the initial temperatures of ambient air are given as $-15 ^{\circ}$C, $5 ^{\circ}$C, and $25 ^{\circ}$C, respectively. The initial temperature of the water droplets is $25 ^{\circ}$C. The dimensionless numbers are $Pr$ = 7.1, $Bo$ = 0.086, $We$ = 3.4 $\times{10}^{-7}$, and $St$ = 0.31. 

The results indicate that air temperature $T_a$ primarily influences the early stages of the droplet icing process. As shown in Fig.~\ref{Ta}, lower $T_a$ helps the supercooled substrate to cool the water droplets, leading to an earlier onset of icing and consequently a faster icing rate (see $t = 0 - 0.6$ s of Fig.~\ref{Ta}). For $t > 0.6$ s, the influence of the air becomes less significant, as different initial air temperatures are cooled to similar levels by the substrate cooling surface. The icing process is then mainly determined by the temperature of the supercooled substrates and the final profiles of the ice droplet are the same though the case with lower initial $T_a$ can reach its final state slightly faster.

\subsection{Three-dimensional icing simulation and experiments}
\begin{table*}[ht!]
    \centering
     \caption{Setup and computational costs of three-dimensional icing simulations}
    \resizebox{\textwidth}{!}{ %
    \begin{tabular}{ccccccc}
        \toprule  
        Methods & \makecell[c]{Grids\\number} & \makecell[c]{CPU\\cores} & \makecell[c]{Time step\\(s)} & Steps & \makecell[c]{Calculation\\time (hour)}& \makecell[c]{Time/step\\(s/step)} \\ 
        \midrule  
        \makecell[c]{Proposed phase-field\\method} & 100 $\times$100$\times$100 & 50 & $1.2\times10^{-6}$ & $8.65\times10^{6}$ & 117.2 & 0.049\\
        \midrule
        \makecell[c]{VoF method in\\Ansys Fluent} &100 $\times$100$\times$100 &$50$ & $4.0\times10^{-6}$& $1.375\times10^{5}$ & 29.5 & 0.77 \\
        \bottomrule  
    \end{tabular} 
    } %
    \label{table costs}
\end{table*}
\begin{figure}[ht!]
    \centering 
    \includegraphics[width=0.85\textwidth]{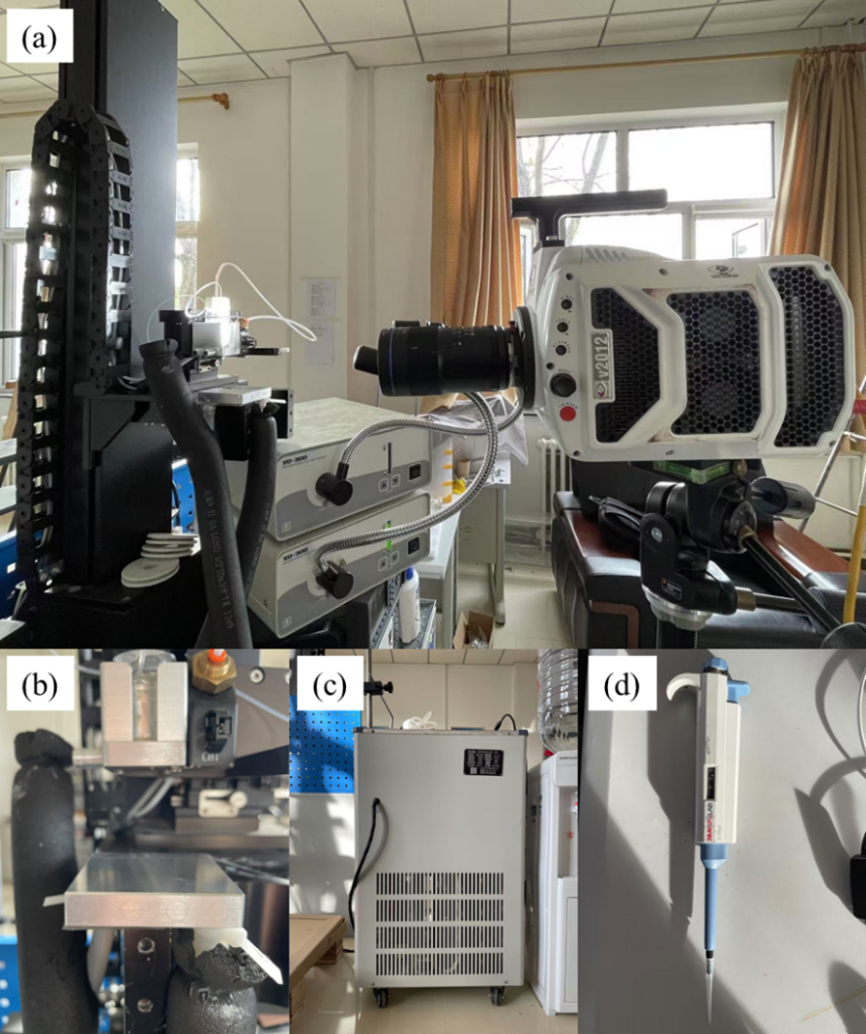}
    \caption{Experimental apparatus. (a) High-speed camera, (b) supercooled substrate, (c) low-temperature thermostatic reaction bath, and (d) droplet generator.} %
\label{experiments}
\end{figure}

\begin{figure}[ht!]
    \centering 
    \includegraphics[width=0.99\textwidth]{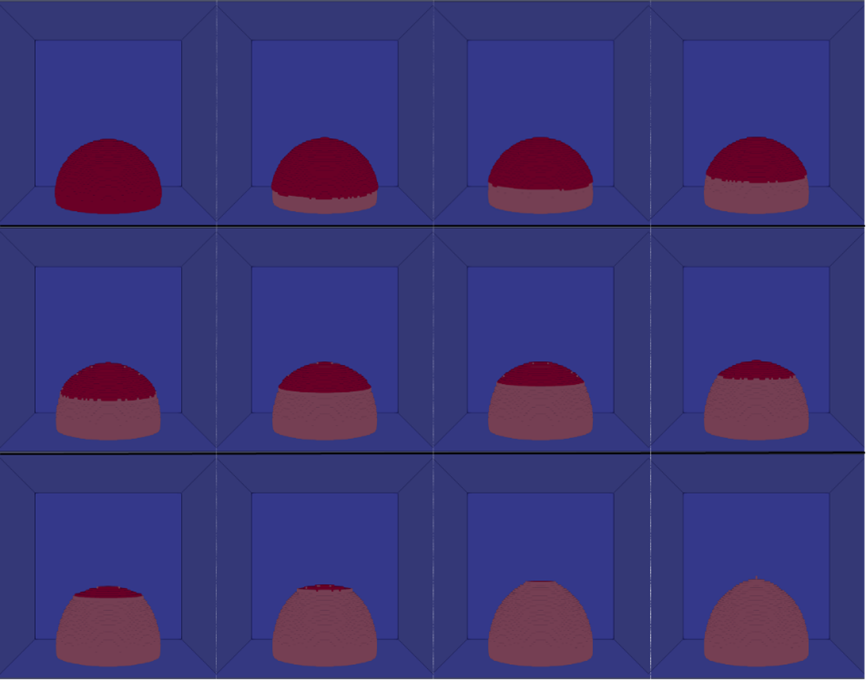}
    \caption{Numerical results of the three-dimensional droplet icing obtained by the proposed phase-field model. The initial shape of the water droplet is approximated as a spherical cap with a radius of 1.08 mm and a height of 1.33 mm. The substrate temperature is -25.3°C and the initial temperature of the water droplet is 7.4°C. From left to right and top to bottom, the times are t = 0, 1, 2, 3, 4, 5, 6, 7, 8, 9, 10, and 10.4 s.} %
\label{3d icing}
\end{figure}
\begin{figure}[ht!]
    \centering 
    \includegraphics[width=0.99\textwidth]{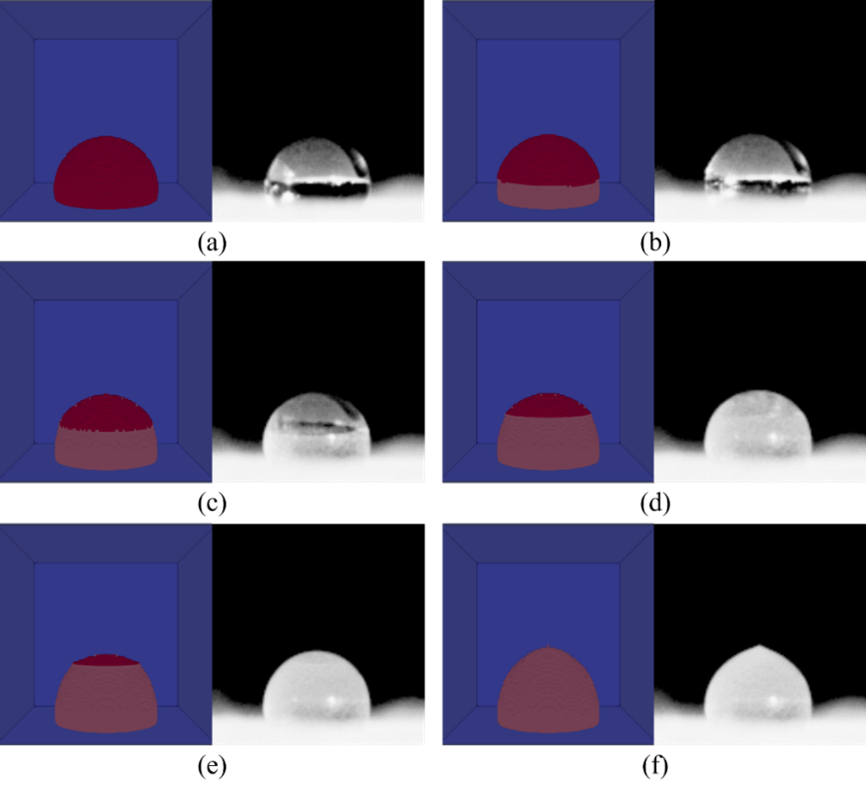}
    \caption{Comparison of the numerical results with the experimental ones. The images on the left of (a)-(f) show the numerical results and those on the right show the experimental ones. (a)-(e) are the snapshots at t = 0, 2, 4, 6, and 8 s, respectively. (f) is the final state of the droplet (t = 10.4 s for the numerical result and t = 10.0 s for the experimental result).} %
\label{comp_with_experiments}
\end{figure}
\begin{figure}[ht!]
    \centering 
    \includegraphics[width=0.99\textwidth]{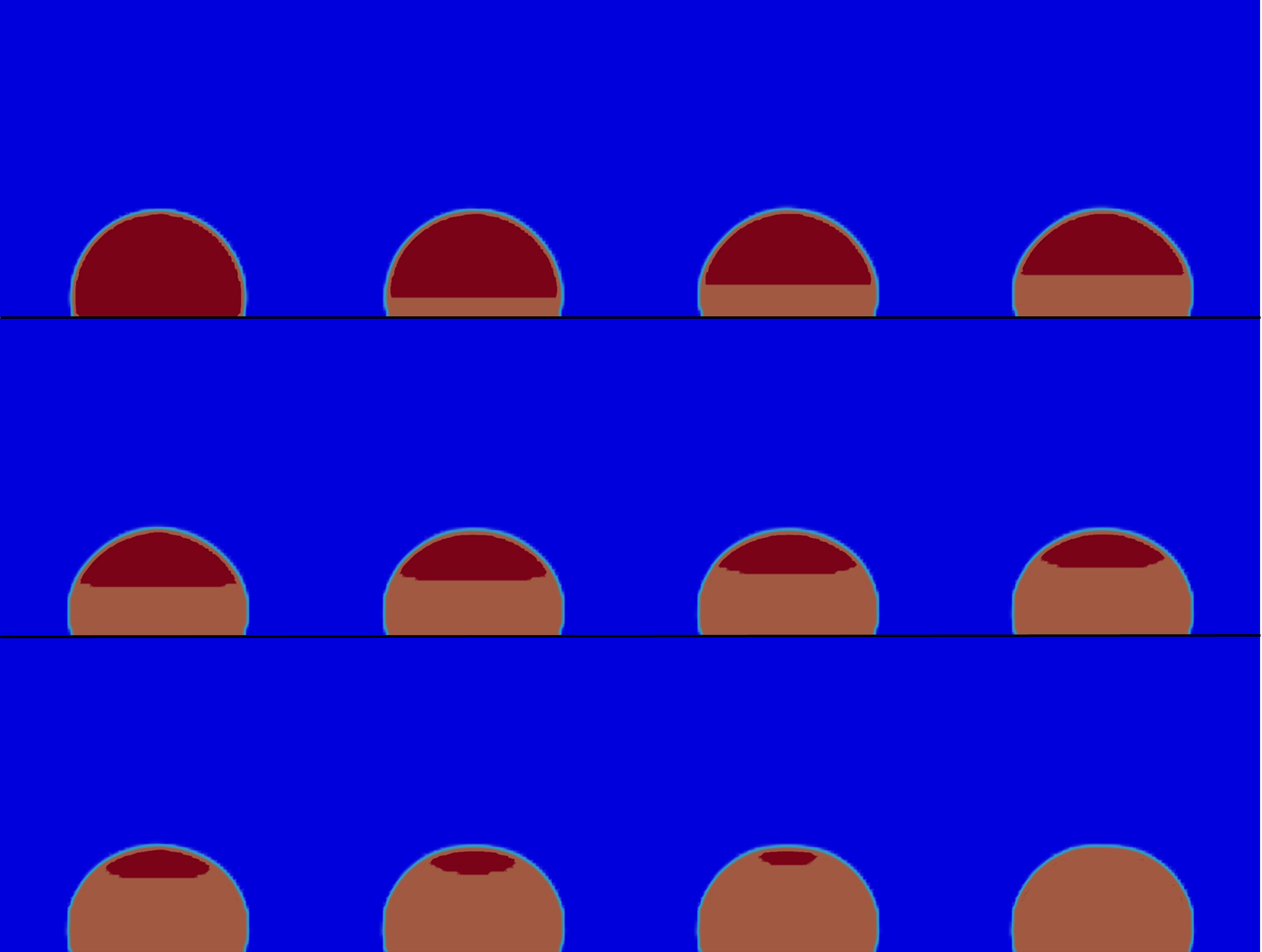}
    \caption{Numerical results of the droplet icing obtained by the VoF method in Ansys Fluent. The initial shape of the water droplet is approximated as a spherical cap with a radius of 1.08 mm and a height of 1.33 mm. The substrate temperature is -25.3°C and the initial temperature of the water droplet is 7.4°C. From left to right and top to bottom, the times are t = 0, 0.05, 0.1, 0.15, 0.2, 0.25, 0.3, 0.35, 0.4, 0.45, 0.5, and 0.55 s.} %
\label{3d_icing_vof}
\end{figure}

Three-dimensional droplet icing processes are simulated in this section following the previous numerical simulations of droplet icing in the Hele-Shaw cell. The results are first compared with laboratory experiments, and subsequently with the Volume-of-Fluid (VoF) method implemented in Ansys Fluent.

Fig.~\ref{experiments} shows the experimental setup. A high-speed camera (Fig.~\ref{experiments}(a)) is used to record the icing process on the supercooled substrate (Fig.~\ref{experiments}(b)). The temperature of the substrate is controlled by a low-temperature thermostatic reaction bath (Fig.~\ref{experiments}(c)) and the droplet generator (Fig.~\ref{experiments}(d)) is used to create millimetre-sized droplets.

The substrate temperature is $-25.3 ^{\circ}$C, and the initial temperature of the water droplet is $7.4 ^{\circ}$C. The initial shape of the water droplet is approximated as a spherical cap with a radius of 1.08 mm and a height of 1.33 mm. The dimensionless numbers are $Pr$ = 7.1, $Bo$ = 0.16, $We$ = 2.5 $\times{10}^{-7}$, and $St$ = 0.32.

For the numerical simulation, 100 $\times$ 100 $\times$ 100 grids are used, with a time step of 1.2 $\times$ 10$^{-6}$ seconds for the proposed phase-field method and 4.0 $\times$ 10$^{-6}$ seconds for the VoF method implemented in Ansys Fluent, respectively. The specific setup and computational costs of the proposed phase-field method and the Fluent VoF method are summarized in Table \ref{table costs}.

Fig.~\ref{3d icing} presents the snapshots of the three-dimensional droplet icing process obtained by the proposed phase-field model. The comparison of the numerical results with the experimental data is presented in Fig.~\ref{comp_with_experiments}. Snapshots at different times show that the numerical results agree well with the experiments. Pointy tips are formed at the final stages of both approaches. The solidification times are close: 10.4 seconds for the numerical results and 10 seconds for the experiments. See the videos of the icing process for both the numerical results and the experiments in the supplemental material.

The numerical results are finally compared with those obtained by the Fluent VoF method (see Fig.~\ref{3d_icing_vof}). The results reveal several differences:
\begin{enumerate}
    \item In the phase-field results, the volume expansion primarily occurs in the vertical direction, whereas in the Fluent VoF results, the expansion is predominantly towards the sides.\\
    \item The experimentally observed pointy tip for the frozen droplet is predicted by the proposed phase-field model but not the Fluent VoF method.\\
    \item  The successful simulation of the solidification process and the pointy tip by the proposed method comes with more computational cost compared against the Fluent VoF method.
\end{enumerate}

\section{Conclusions} \label{5}
In this paper, we establish a three-phase AC-CH-NS phase-field model by combining the NS equations and the energy equation with the AC and CH equations. The proposed model can fully utilize the advantages of both the AC equation and the CH equation, e.g., the AC equation can take the latent heat release into account during the phase-changing process while the CH equation is volume conserved in the whole domain. This so-called AC-CH-NS scheme can be further extended to AC-AC-NS or CH-CH-NS schemes, depending on the properties of the corresponding multiphase flows.

We then establish the temporal and spatial discretization schemes for the proposed AC-CH-NS model in detail with the finite difference method. This scheme is then proved to satisfy the energy dissipation law, which is an important feature of the phase-field model since it demonstrates the model is unconditionally stable. The density and viscosity in our scheme are variables, resulting in challenges to prove the energy stability of the corresponding numerical scheme \cite{shen2012modeling}. In the previous efforts, the energy stability of the two-phase phase-field models, e.g., AC-NS model and CH-NS model, has been well studied. However, for the three-phase flows, to our knowledge, it is the first time to establish an energy-stable AC-CH-NS discrete scheme with varying density and viscosity.

By employing the established phase-field model, the water droplet cooled by supercooled substrates is numerically simulated and the pointy tips are well obtained after solidification. We note that in many of the previous numerical studies, the pointy tips were not well captured, though these singularity tips are obvious in nature and experiments. Our results demonstrate that the pointy tips are formed due to the upward expansion during the icing process. Specifically, the cooling of the water droplet's surface by the surrounding air leads to the formation of a thin ice layer. This ice layer restricts lateral expansion of the water, causing the droplet to expand predominantly upwards as it freezes, thereby resulting in the formation of pointed tips. The parameter studies on the temperatures of both the substrates and the initial surrounding air show that $T_{sub}$ mainly influences the icing speed while $T_a$ primarily affects the icing process in the very early stages with the final shapes and the pointy tips of the droplet remaining the same. Furthermore, the icing front at the upper hemisphere of the droplet exhibits concave shapes due to the difference in the thermal diffusivity between air and water.

Experiments have been conducted to validate the numerical results obtained from the proposed phase-field model. The numerical simulations agree well with the experimental observations.

Overall, our work contributes to the fundamental properties of the phase-field method, i.e., the energy stability, and paves the way for further studies on the applications of the phase-field model, e.g., crystal growth, solidification in alloys, and other kinds of multiphase flows.
\section*{Acknowledgments}
The authors gratefully acknowledge financial support from the China Scholarship Council, National Natural Science Foundation of China
(Grant No. 92152109, and Grant No. 11988102), the ARCHER2-eCSE05-6 project, and the EU H2020-MSCA-RISE-2017 project CTFF (Control of Turbulent Friction Force) under the grant agreement number 777717.


\appendix
\section{Proof of the discrete energy law}
\label{app}

The proof of discrete energy law shown in Theorem~\ref{theorem 3} is as follows:
\begin{proof}
By multiplying Eq.~(\ref{ac_time}) with $2(c^{n+1}-c^n)$ and taking the $L^2$-inner product, we get
\begin{equation}
\begin{aligned}
 &2\Delta t (\frac{c^{n+1}-c^n}{\Delta t} + \boldsymbol{u}^{n+1}\cdot \nabla c^n, \frac{c^{n+1}-c^n}{\Delta t}) + 2S_1 \norm {c^{n+1}-c^n}^2 \\
 & = - M_c \xi_c ^2 (\norm {\nabla c^{n+1}} ^2 -\norm {\nabla c^{n}}^2  + \norm {\nabla c^{n+1} -\nabla c^{n}} ^2)- 2 M_c (F'(c^n), c^{n+1}-c^n),
 \label{ac_time_l2}
\end{aligned}
\end{equation}
where the property $(b-a,2b) = \abs b ^2 - \abs a ^2 + \abs{b-a}^2$ is used.

The last term in Eq.~(\ref{ac_time_l2}) can be rewritten following the Taylor expansion:
\begin{equation}
      F'(c^{n})(c^{n+1}- c^n)
      = F(c^{n+1}) -  F(c^{n})  
      - \frac{F''(\epsilon_1)}{2} (c^{n+1}- c^n)^2,
      \label{taylor1}
\end{equation}
where $\epsilon_1$ is a value of $c$ in $[-1,0]$ and $F''(c)$ denotes the second derivative of $F(c)$ with respect to $c$:
\begin{equation}
    F''(c) = 6c^2+6c+1 \leq 1 \ \ \ \  {\rm for}\  c \in [-1,0].
    \label{f2c}
\end{equation}

By multiplying Eq.~(\ref{ch_time}) with $2 \Delta t \mu_\phi^{n+1}$ and taking the $L^2$-inner product, we get
\begin{equation}
   2 (\phi^{n+1}-\phi^n, \mu^{n+1}_\phi) + 2 \Delta t (\boldsymbol{u}^{n+1}\cdot \nabla \phi ^n, \mu_\phi^{n+1}) = -2 M_\phi \Delta t \norm {\nabla \mu_\phi^{n+1}}^2.
    \label{ch_time_l2}
\end{equation}

By multiplying Eq.~(\ref{muphi_time}) with $2(\phi^{n+1} - \phi^{n})$ and taking the $L^2$-inner product, we get
\begin{equation}
\begin{aligned}
    2(\mu_\phi^{n+1}, \phi^{n+1} - \phi^{n}) &- 2S_2 \norm {\phi^{n+1}-\phi^n}^2 \\
   & = B_2 \xi^2_\phi (\norm {\nabla \phi^{n+1}}^2 - \norm {\nabla \phi^{n}}^2 + \norm {\nabla \phi^{n+1}- \nabla \phi^{n}}^2) \\
   &+ 2 B_2 (G'(\phi^n), \phi^{n+1}- \phi^{n}).
    \label{muphi_time_l2}
\end{aligned}
\end{equation}

The last term in Eq.~(\ref{muphi_time_l2}) can be rewritten following the Taylor expansion:
\begin{equation}
      G'(\phi^{n})(\phi^{n+1}- \phi^n)
      = G(\phi^{n+1}) -  G(\phi^{n})  
      - \frac{G''(\epsilon_2)}{2} (\phi^{n+1}- \phi^n)^2,
      \label{taylor2}
\end{equation}
where $\epsilon_2$ is a value of $\phi$ in $[-1,1]$ and $G''(\phi)$ denotes the second derivative of $G(\phi)$ with respect to $\phi$:
\begin{equation}
    G''(\phi) = 3\phi^2 -1 \leq 2 \ \ \ \  {\rm for}\  \phi \in [-1,1].
    \label{g2phi}
\end{equation}

Taking the $L^2$-inner product of Eq.~(\ref{NS1_time}) with $2 \Delta t \boldsymbol{u}^{n+1}$, we have
\begin{equation}
    \begin{aligned}
  & \norm{\sqrt{\rho^{n+1}}\boldsymbol{u}^{n+1}}^2
    - \norm{\sqrt{\rho^{n}}\boldsymbol{u}^{n}} ^2
    + \norm{\sqrt{\rho^{n}}(\boldsymbol{u}^{n+1}-\boldsymbol{u}^{n})}^2 \\
   & = -2\Delta t (p^{n+1} - 2 p^n + p^{n-1}, \nabla \cdot \boldsymbol{u}^{n+1})
    +  2\Delta t (p^{n+1}, \nabla \cdot \boldsymbol{u}^{n+1}) \\
   & - 2\Delta t \norm{\sqrt{\mu^{n+1}}D(\boldsymbol{u}^{n+1})}^2
    - 2\Delta t\frac{B_1}{M_c} ((\frac{c^{n+1}-c^n}{\Delta t} + \boldsymbol{u}^{n+1} \cdot \nabla c^n) \nabla c^n, \boldsymbol{u}^{n+1}) \\
   & +  2\Delta t(\mu_\phi ^{n+1} \nabla \phi ^n, \boldsymbol{u}^{n+1})
   +  2\Delta t (\boldsymbol{F}_b ^n, \boldsymbol{u}^{n+1}),
        \label{NS1_time_l2}
    \end{aligned}
\end{equation}
where the properties
\begin{equation}
\begin{aligned}
      (\frac{1}{2}(\rho^{n+1}+\rho^{n})\boldsymbol{u}^{n+1}-\rho^n \boldsymbol{u}^{n}, 2\boldsymbol{u}^{n+1})
   &= \norm{\sqrt{\rho^{n+1}}\boldsymbol{u}^{n+1}}^2
    - \norm{\sqrt{\rho^{n}}\boldsymbol{u}^{n}} ^2 \\
   & + \norm{\sqrt{\rho^{n}}(\boldsymbol{u}^{n+1}-\boldsymbol{u}^{n})}^2,
\end{aligned}
\end{equation}
and
\begin{equation}
    \int_\Omega (\rho\boldsymbol{u_1}\cdot \nabla)\boldsymbol{u_2} \cdot \boldsymbol{u_2}d \boldsymbol {x}
+ \frac{1}{2} \int_\Omega \nabla\cdot (\rho\boldsymbol{u_1})\boldsymbol{u_2}\cdot \boldsymbol{u_2}d \boldsymbol {x} =0,\ {\rm if} \ \boldsymbol{u_1}\cdot \boldsymbol n |_{\partial \Omega}=0,
\end{equation}
are used with $\boldsymbol{u_1}$, and $\boldsymbol{u_2}$ representing two velocity vectors \cite{shen2012modeling}.

Taking the $L^2$-inner product of Eq.~(\ref{NS2_time}) with $2\frac{ (\Delta t)^2}{\bar{\rho}}(p^{n+1} - 2 p^n + p^{n-1})$, we have
\begin{equation}
\begin{aligned}
       - \frac{(\Delta t)^2}{\bar{\rho}} (\norm{\nabla (p^{n+1} - p^n)}^2 - \norm{\nabla (p^{n} - p^{n-1})}^2 + \norm{\nabla (p^{n+1}-2p^{n} + p^{n-1})}^2 )\\
    =2 \Delta t (\nabla \cdot \boldsymbol{u}^{n+1}, p^{n+1} - 2 p^n + p^{n-1}).
    \label{NS2_time_l2_1}
\end{aligned}
\end{equation}

By multiplying Eq.~(\ref{NS2_time}) with $2\frac{(\Delta t)^2}{\bar{\rho}}p^{n+1}$ and taking the $L^2$-inner product, we get

\begin{equation}
\begin{aligned}
   - \frac{(\Delta t)^2}{\bar{\rho}} (\norm{\nabla p^{n+1}}^2 - \norm{\nabla p^{n}}^2  + \norm {\nabla (p^{n+1} - p^n)}^2 ) = 2 \Delta t (\nabla \cdot \boldsymbol{u}^{n+1}, p^{n+1}).
     \label{NS2_time_l2_2}
\end{aligned}
\end{equation}

Taking the difference of Eq.~(\ref{NS2_time}) at time steps $n+1$ and $n$, we obtain
\begin{equation}
    \nabla(p^{n+1}-2p^n +p^{n-1} ) = \frac{\bar{\rho}}{\Delta t} (\boldsymbol{u}^{n+1} -\boldsymbol{u}^{n}).
    \label{NS2_time_n+1_n}
\end{equation}

Taking the inner product of both sides of the equation themselves, we obtain
\begin{equation}
   \norm { \nabla(p^{n+1}-2p^n +p^{n-1} ) }^2 = \frac{\bar{\rho}^2}{(\Delta t)^2} \norm {(\boldsymbol{u}^{n+1} -\boldsymbol{u}^{n})}^2.
    \label{NS2_time_n+1_n_squre}
\end{equation}

By substituting Eq.~(\ref{NS2_time_l2_1}), Eq.~(\ref{NS2_time_l2_2}), and Eq.~(\ref{NS2_time_n+1_n_squre}) into Eq.~(\ref{NS1_time_l2}), we have
\begin{equation}
    \begin{aligned}
  & \norm{\sqrt{\rho^{n+1}}\boldsymbol{u}^{n+1}}^2
    - \norm{\sqrt{\rho^{n}}\boldsymbol{u}^{n}} ^2
    + (\rho^{n}- \bar{\rho}) \norm{(\boldsymbol{u}^{n+1}-\boldsymbol{u}^{n})}^2 \\
    & = \frac{(\Delta t)^2}{\bar{\rho}} (- \norm{\nabla(p^n-p^{n-1})}^2 -    \norm{\nabla p^{n+1}}^2 + \norm{\nabla p^n}^2)\\
   & - 2\Delta t \norm{\sqrt{\mu^{n+1}}D(\boldsymbol{u}^{n+1})}^2
     - 2\Delta t\frac{B_1}{M_c}((\frac{c^{n+1}-c^n}{\Delta t} + \boldsymbol{u}^{n+1} \cdot \nabla c^n) \nabla c^n, \boldsymbol{u}^{n+1}) \\
   & +  2\Delta t(\mu_\phi ^{n+1} \nabla \phi ^n, \boldsymbol{u}^{n+1})
   +  2\Delta t (\boldsymbol{F}_b ^n, \boldsymbol{u}^{n+1}).
        \label{NS1_time_l2_4}
    \end{aligned}
\end{equation}

By combining Eq.~(\ref{NS1_time_l2_4}) with $-\frac{B_1}{M_c}$ timing Eq.~(\ref{ac_time_l2}), Eq.~(\ref{taylor1}), Eq.~(\ref{ch_time_l2}), Eq.~(\ref{muphi_time_l2}), and Eq.~(\ref{taylor2}), we have
\begin{equation}
    \begin{aligned}
  & \norm{\sqrt{\rho^{n+1}}\boldsymbol{u}^{n+1}}^2
    - \norm{\sqrt{\rho^{n}}\boldsymbol{u}^{n}} ^2
    + (\rho^{n}- \bar{\rho}) \norm{(\boldsymbol{u}^{n+1}-\boldsymbol{u}^{n})}^2 \\
    & = \frac{(\Delta t)^2}{\bar{\rho}} (- \norm{\nabla(p^n-p^{n-1})}^2 -    \norm{\nabla p^{n+1}}^2 + \norm{\nabla p^n}^2)\\
   & - 2\Delta t \norm{\sqrt{\mu^{n+1}}D(\boldsymbol{u}^{n+1})}^2
    - 2 \Delta t\frac{B_1}{M_c} \norm{\frac{c^{n+1}-c^n}{\Delta t} + \boldsymbol{u}^{n+1} \cdot \nabla c^n }^2 \\
    & -2 S_1\frac{B_1}{M_c} \norm{c^{n+1}-c^n}^2
    - B_1 \xi^2_c (\norm {\nabla c^{n+1}} ^2 -\norm {\nabla c^{n}}^2  + \norm {\nabla c^{n+1} -\nabla c^{n}} ^2)\\
   &- 2 B_1 (F(c^{n+1})-F(c^{n}), 1) 
   + 2 B_1 (\frac{F''(\epsilon_1^n)}{2}(c^{n+1}- c^n)^2, 1) \\
   & -2 M_\phi \Delta t \norm {\nabla \mu_\phi^{n+1}}^2 
   -2 S_2 \norm {\phi^{n+1}-\phi^n}^2 \\
   &- B_2 \xi^2_\phi (\norm {\nabla \phi^{n+1}}^2 - \norm {\nabla \phi^{n}}^2
   + \norm {\nabla \phi^{n+1}- \nabla \phi^{n}}^2)\\
   &-2 B_2 (G(\phi^{n+1})-G(\phi^{n}), 1)
    +2 B_2 (\frac{G''(\epsilon_2^n)}{2}(\phi^{n+1}- \phi^n)^2, 1) \\
   &-2 (E^{n+1}_g-E^{n}_g, 1),
    \label{NS1_time_l2_5}
    \end{aligned}
\end{equation}
where $2\Delta t (\boldsymbol{F}_b ^n, \boldsymbol{u}^{n+1})= -2(E^{n+1}_g-E^{n}_g, 1)$ is used.

Taking the properties that
\begin{equation}
\begin{aligned}
    &    2 B_1 (\frac{F''(\epsilon_1^n)}{2}(c^{n+1}- c^n)^2, 1) 
    - 2 S_1 \frac{B_1}{M_c} \norm{c^{n+1}-c^n}^2\\
   & \leq 2 B_1 (\frac{1}{2}-\frac{S_1}{M_c}) \norm{c^{n+1}-c^n}^2\\
   & \leq 0,
\end{aligned}
\end{equation}
and 
\begin{equation}
\begin{aligned}
    &2 B_2 (\frac{G''(\epsilon_2^n)}{2}(\phi^{n+1}- \phi^n)^2, 1)
    -2 S_2 \norm {\phi^{n+1}-\phi^n}^2 \\
    & \leq 2(B_2 - S_2) \norm {\phi^{n+1}-\phi^n}^2\\
     & \leq 0,
\end{aligned}
\end{equation}
and
\begin{equation}
    (\rho^{n}- \bar{\rho}) \norm{(\boldsymbol{u}^{n+1}-\boldsymbol{u}^{n})}^2 \geq 0,
\end{equation}
where $F''(c)\leq 1 $ (see Eq.~(\ref{f2c})), $G''(\phi)\leq 2 $ (see Eq.~(\ref{g2phi})), $S_1 \geq \frac{1}{2} M_c$, $S_2 \geq B_2$, and $\rho^{n} \geq \bar{\rho}$ are used,
Eq.~(\ref{NS1_time_l2_5}) can be finally simplified as follows
\begin{equation}
    \begin{aligned}
  & \norm{\sqrt{\rho^{n+1}}\boldsymbol{u}^{n+1}}^2
    - \norm{\sqrt{\rho^{n}}\boldsymbol{u}^{n}} ^2
    + \frac{(\Delta t)^2}{\bar{\rho}} \norm{\nabla (p^n- p^{n-1})}^2 \\
  & + \frac{(\Delta t)^2}{\bar{\rho}} (\norm{\nabla p^{n+1}}^2 - \norm{\nabla p^n}^2) 
  + 2 \Delta t \norm{\sqrt{\mu^{n+1}}D(\boldsymbol{u}^{n+1})}^2 \\
  & + 2 \Delta t\frac{B_1}{M_c} \norm{\frac{c^{n+1}-c^n}{\Delta t} + \boldsymbol{u}^{n+1} \cdot \nabla c^n }^2 \\
  &+ B_1 \xi^2_c (\norm {\nabla c^{n+1}} ^2 -\norm {\nabla c^{n}}^2  + \norm {\nabla c^{n+1} -\nabla c^{n}} ^2)
  + 2 B_1 (F(c^{n+1})-F(c^{n}), 1) \\ 
 & + 2 M_\phi \Delta t \norm {\nabla \mu_\phi^{n+1}}^2
   + B_2 \xi^2_{\phi}(\norm {\nabla \phi^{n+1}}^2 - \norm {\nabla \phi^{n}}^2 + \norm {\nabla \phi^{n+1}- \nabla \phi^{n}}^2) \\
  &+ 2 B_2 (G(\phi^{n+1})-G(\phi^{n}), 1)  
  +2 (E^{n+1}_g-E^{n}_g, 1)\\
  & \leq 0.
    \label{NS1_time_l2_6}
    \end{aligned}
\end{equation}

The above inequality completes the proof.
\end{proof}



\bibliographystyle{elsarticle-num} 
\bibliography{refs}





\end{document}